\documentclass[twocolumn,journal]{IEEEtran}
\IEEEoverridecommandlockouts
\usepackage{tikz}
\usepackage[section]{placeins}
\usepackage{bm}
\usepackage{bbm}
\usepackage{subcaption}
\usepackage{float}
\usepackage{amsmath}
\usepackage{graphicx}
\usepackage{epstopdf}
\usepackage{amssymb}
\allowdisplaybreaks
\usepackage{amssymb,amsmath}
\usepackage{amsgen,amsfonts,amsbsy,amsthm}
\usepackage{latexsym}
\usepackage{times}
\usepackage{algorithm}
\usepackage[noend]{algpseudocode}
\makeatletter
\def\BState{\State\hskip-\ALG@thistlm}
\makeatother
\DeclareMathOperator*{\argmin}{arg\,min}
\DeclareMathOperator*{\argmax}{arg\,max}
\newcommand{\curly}[1]{\left\{#1\right\}}
\newcommand{\spn}[1]{\mbox{\texttt{span}}\left(#1\right)}
\newcommand{\inprod}[2]{\left\langle #1,#2 \right\rangle}
\newcommand{\bvec}[1]{\bm{#1}}
\newcommand{\real}{\mathbb{R}}
\newcommand{\complex}{\mathbb{C}}

\newcommand{\supp}{\texttt{supp}}
\newcommand{\proj}[1]{\mathbf{P}_{#1}}
\newcommand{\dualproj}[1]{\mathbf{P}_{#1}^\perp}
\newcommand{\norm}[1]{\left\|#1\right\|_2}
\newcommand{\opnorm}[2]{\left\|#1\right\|_{#2}}
\newcommand{\abs}[1]{\left|#1\right|}
\newcommand{\sign}[1]{\texttt{csgn}\left(#1\right)}

\newcommand{\prob}[1]{\mathbb{P}\left(#1\right)}
\newcommand{\expect}[2]{\mathbb{E}_{#1}\left(#2\right)}
\newcommand{\indicator}[1]{\mathbbm{1}\left\{#1\right\}}
\newtheorem{lem}{Lemma}[section]
\newtheorem{thm}{Theorem}[section]

\newtheorem{prop}{Proposition}[section]
\newtheorem{define}{Definition}[section]

\title{A Two Stage Generalized Block Orthogonal Matching Pursuit (TSGBOMP) Algorithm}
\begin{document}
\setlength\abovedisplayskip{0pt}
\setlength\belowdisplayskip{0pt}
\vspace{-5mm}
\author{Samrat Mukhopadhyay$^1$, {\sl Student Member, IEEE}, and Mrityunjoy
Chakraborty$^2$, {\sl Senior Member, IEEE}

\thanks{The authors are with the department of Electronics and Electrical Communication
Engineering, Indian Institute of Technology, Kharagpur, INDIA (email : $^1$samratphysics@gmail.com, $^2$mrityun@ece.iitkgp.ernet.in).
}}
\IEEEoverridecommandlockouts
\maketitle
\begin{abstract}
	 Recovery of an unknown sparse signal from a few of its projections is the key objective of compressed sensing. 
	Often one comes across signals that are not ordinarily sparse but are sparse blockwise. Existing block sparse recovery algorithms like BOMP make the
	assumption of uniform block size and known block boundaries, which are, however, not very practical in many applications. This paper addresses this problem and proposes a two step procedure,
	where the first stage is a coarse block location identification stage while the second stage carries out finer localization of a non-zero cluster within the window selected in the first stage.
	A detailed convergence analysis of the proposed algorithm is carried out by first defining  the so-called pseudoblock-interleaved block RIP of the given generalized block sparse signal and then imposing upper bounds on the corresponding RIC. We also extend the analysis for complex vector as well as matrix entries where it turns out that the extension is non-trivial and requires special care. Furthermore, assuming real Gaussian sensing matrix entries, we find a lower bound on the probability that the derived recovery bounds are satisfied. The lower bound suggests that there are sets of parameters such that the derived bound is satisfied with high probability. Simulation results confirm significantly improved performance of the proposed algorithm as compared to BOMP.
\end{abstract}
\section{Introduction}
\label{sec:intro}
The problem of compressed sensing (CS) has emerged as a powerful tool to retrieve an unknown \emph{sparse} vector with a few nonzero entries in some unknown coordinates, from a small set of measurements, obtained using a sensing matrix~\cite{candes2006robust}. The recovery problem in CS is formulated as an optimization problem with a set of linear equations as constraint and a cost function, that measures some property of the unknown vector, as the objective~\cite{candes_decoding_2005}. Often such optimization problems involve nonconvex objective functions, which makes the recovery problem NP hard in general~\cite{candes2006robust}. A major line of research in the CS literature attacks this problem by replacing the nonconvex objective function with a convex objective function, and then study and analyze the performance of the algorithm by finding conditions on the measurement matrix under which the solution to the convexified problem coincides with the exact solution to the original nonconvex problem~\cite{candes-tao-stable-recovery}~\cite{cai2009recovery}. However, as convex optimization approaches are often computationally too expensive~\cite{donoho2012sparse}, a large number of greedy heuristic approaches, like matching pursuit~\cite{mallat_matching_1993}, Orthogonal Matching Pursuit (OMP)~\cite{pati1993orthogonal,tropp2004greed}, to name a few, have been proposed to address the CS problem. Among these, the OMP algorithm has attracted a lot of attention because of its simplicity and capability of good recovery performance with low computational complexity. The OMP algorithm proceeds by gradually constructing the support of the unknown signal by iteratively updating a support by appending indices that correspond to maximum absolute correlations between a certain residual vector and the columns of the sensing matrix. Thereafter, the residual vector is updated by finding the orthogonal projection error found after projecting the measurement vector over the vector space spanned by the vectors with the support constructed so far.

In many practical applications like multiband signal processing~\cite{mishali2009blind,mishali2010theory} and multiple measurement vector (MMV) recovery problem~\cite{cotter2005sparse,chen2005sparse} signals typically have the \emph{block sparse} structure~\cite{eldar2009robust,parvaresh2008recovering}. In this structure, the nonzero elements tend to occur in clusters of known size and it is generally known that such clusters are located within a few prespecified blocks known to the end user. The compressed sensing recovery problem with the block sparse structure was studied in detail in~\cite{eldar2009robust}. Furthermore, Block OMP (BOMP) was proposed~\cite{eldar2010block} as an extension of OMP for the block-sparse recovery problem from compressed measurements. The BOMP algorithm works similar to the OMP algorithm. The main difference is that instead of taking correlations between the residual and each column of the sensing matrix, the BOMP algorithm first forms a vector of correlations between a block of columns and the residual vector, and then finds the block for which the norm of such a vector is maximized. Then the residual is updated similar to OMP by taking the orthogonal projection error after projecting the measurement vector on the space spanned by the columns of the blocks identified so far. All these papers analyzed the block recovery problem using methods like block-coherence~\cite{eldar2010block}. Recently, a block restricted isometry property(BRIP)~\cite{candes_decoding_2005,baraniuk2008simple,eldar2009robust} based recovery analysis~\cite{wen2018sharp} has found conditions on the block restricted isometry constant(BRIC) of the sensing matrix to ensure perfect recovery using the BOMP algorithm.

In many applications, like the atomic decomposition of audio signals~\cite{gribonval2003harmonic} the exact block partitions of the unknown vector is not beforehand. Although a few algorithms have been proposed to address the recovery of this kind of signals~\cite{zhang2013extension,fang2015pattern}, all of them use the Bayesian learning framework, which impose prior distributional assumptions on the unknown vector. In this paper we propose and study a new non-Bayesian algorithm called the two stage generalized block OMP (TSGBOMP), which has similar structure to the BOMP algorithm except that the block identification is performed in two stages. The first stage is a coarse block location identification stage where, similar to the BOMP algorithm, among a prespecified set of \emph{windows} (i.e. sets of consecutive columns that the whole set of columns is partitioned into) an window of columns having the maximum correlation (with some prior residual) is selected. In the second stage, the algorithm conducts a finer search for a block by calculating the correlations (with some prior residual) corresponding to all overlapping consecutive clusters of columns throughout the window selected, and then finds the one cluster having the largest absolute correlation \footnote{This philosophy of block selection in the second stage is inspired from a recent paper~\cite{kannu2018spcom} which searches for a block by calculating absolute correlations corresponding to all possible overlapping clusters of columns in the matrix with certain residual vector and then selecting the one having the highest value.
}.

Our main contributions are the following: 1) We propose a new recovery algorithm called TSGBOMP which which uses a two-stage strategy for recovering generalized block sparse vectors with no prior knowledge of block partitions. 2) The analysis of TSGBOMP necessitates the introduction of a new kind of RIP tailored to the particular structure of the unknown vector, termed as pseudoblock-interleaved block RIP (PIBRIP). This kind of RIP is motivated by the model-RIP introduced in~\cite{baraniuk2010model} for analyzing signals with union of subspace structure. 3) We analyze the TSGBOMP algorithm using the PIBRIP property and find recovery condition that ensures the exact recovery using TSGBOMP. 4) We exhibit using Gaussian random matrices that there are matrices which can satisfy the recovery guarantee deduced in the paper with very high probability. 5) Finally, we use numerical simulations to exhibit the superior probability of recovery performances of the TSGBOMP algorithm with respect to the BOMP algorithm in terms of recovering the signal with the generalized block sparse structure.

\section{Notations}
\label{sec:notations}
The following notations have been used throughout the paper : `$H$' in
superscript indicates matrix / vector Hermitian conjugate, $[n]$ denotes the set of indices $\{1,2,\cdots,\ n\}$. For any vector $\bvec{x}\in \complex^n$, the \emph{support} of $\bvec{x}$, denoted by $\supp(\bvec{x})$, is defined as the set of indices corresponding to the nonzero values of $\bvec{x}$, i.e., $\supp(\bvec{x})=\{i\in [n]|[\bvec{x}]_i\ne 0\}$. The symbol $\bm{\phi}_i$ denotes the
$i$ th column of $\bvec{\Phi},\;i \in [n]$ and all the columns of $\bvec{\Phi}$ are assumed to have unit $l_2$ norm, i.e.,
$\|\bm{\phi}_i\|_2=1$, which is a common assumption in literature~\cite{tropp2007signal},~\cite{wang2017recovery}. For any two vectors $\bvec{u}, \bvec{v}\in \complex^n$, $\inprod{\bvec{u}}{\bvec{v}} = \bvec{u}^H \bvec{v}$.
For any $S\subseteq [n]$,
$\bvec{x}_S$ denotes a vector comprising those entries of $\bvec{x}$ that
are indexed by numbers belonging to $S$. Similarly, $\bvec{\Phi}_S$ denotes
the submatrix of $\bvec{\Phi}$ formed with the columns of
$\bvec{\Phi}$ having column numbers given by the index set $S$. For any submatrix $\bvec{A}$ of $\bvec{\Phi}$, define $I(\bvec{A})$ as the set of indices of the columns of $\bvec{\Phi}$ that constitute $\bvec{A}$. We use $\bvec{\Phi}[i]$ to denote the set of $L$ consecutive columns of $\bvec{\Phi}$, with indices $(i - 1) L + 1$ to $Li,\;i=1,\;2,\cdots,,\frac{n}{L}$ (throughout the paper, we assume $n$ to be divisible by $L$), which is hereafter being referred to as the $i^\mathrm{th}$ \emph{window} or the window with index $i$.
Similarly, $\bvec{x}[i]$ denotes the vector comprising of the entries of $\bvec{x}$ that are indexed by $(i-1)L + 1$ to $Li$. We denote by $\bvec{\Phi}[S]$ (resp. $\bvec{x}[S]$) the collection of columns (resp. entries) of $\bvec{\Phi}$ (resp. $\bvec{x}$) corresponding to all the windows with indices in the set $S$. A set of $b(\ge 1)$ consecutive indices is called a \emph{block}. A block is indexed by the starting index of the block, i.e., the $i$-th block starts with the index $i$. In this paper, we consider blocks that are non-overlapping though they can be adjacent. The set of the first indices of the ``true'' nonzero blocks of $\bvec{x}$ \footnote{By ``nonzero block of $\bvec{x}$" is meant a block over which $\bvec{x}$ has non-zero values.} is denoted by $T=\{t_1,t_2,\cdots\}$, where the indices $\{t_i\}_{i\ge 1}$ are recursively defined as follows: $t_1=\min\{j\ge 1|[\bvec{x}]_j\ne 0\}$, and $t_{k+1}=\min\{j\ge t_k+b|[\bvec{x}]_j\ne 0\},\ k\ge 1$.
We denote by $\bvec{\Phi}\curly{i}$ the submatrix $\bvec{\Phi}_Z$ where $Z\;=\;\{i,i+1,\cdots,i+b-1\}$ is the $i^\mathrm{th}$ block. The vector $\bvec{x}\curly{i}$ is defined analogously. For any set $S\subseteq [n]$, $\bvec{\Phi}\curly{S}$ (resp. $\bvec{x}\curly{S}$) denotes the collection of columns (resp. entries) corresponding to the blocks beginning with the indices in set $S$. If $\bvec{\Phi}\curly{S}$
has full column rank of $|S|\times b$ ($|S|\times b<m$), then the Moore-Penrose
pseudo-inverse of $\bvec{\Phi}\curly{S}$ is given by
$\bvec{\Phi}^\dagger\curly{S}=(\bvec{\Phi}^H\curly{S}\bvec{\Phi}\curly{S})^{-1}\bvec{\Phi}^H\curly{S}$.
The matrices $\proj{S}=\bvec{\Phi}\curly{S}\bvec{\Phi}^\dagger\curly{S}$ and $\dualproj{S}=\bvec{I}-\proj{S}$
respectively denote the orthogonal projection operators associated with
$\spn{\bvec{\Phi}\curly{S}}$ and the orthogonal
complement of $\spn{\bvec{\Phi}\curly{S}}$. Finally, for any matrix $\bvec{A}$, we denote by $\opnorm{\bvec{A}}{2\to 2}$ the operator norm of $\bvec{A}$ defined as $\opnorm{\bvec{A}}{2\to 2}=\max_{\bvec{x}\ne \bvec{0}}\frac{\norm{\bvec{A x}}}{\norm{\bvec{x}}}$, and can be shown to be equivalent to $\max_{\bvec{x}\ne \bvec{0}}\frac{\abs{\bvec{x}^H\bvec{A x}}}{\bvec{x}^H\bvec{x}}$ when $\bvec{A}$ is Hermitian~\cite[pp. 519]{foucart2013mathematical}. We use the abbreviation w.l.o.g. for without loss of generality.
\section{Proposed Algorithm}
\label{sec:proposed-algorithm}
The proposed TSGBOMP algorithm aims at recovering an unknown vector $\bvec{x}$ with $K$ nonzero blocks of size $b$ each. However, unlike conventional approaches like the BOMP~\cite{eldar2010block}, it
does not assume the exact block locations to be known a priori. It is assumed that there can be at most $p$ adjacent non-zero blocks in $\bvec{x}$, forming a nonzero \textit{cluster} of maximum size $B (=pb)
$\footnote{In this paper, we use the notion of cluster in most cases rather than block, as the former is more general (a block is a cluster with $p=1$).}.
The nonzero clusters are not contiguous (i.e., they are separated by zeros), and if there are a total of $r$ such nonzero clusters, with the $s^{th}$ cluster having size $j_sb,\ 1\le j_s\le p,\;s=1,2,\cdots,r$
(i.e., it has $j_s$ contiguous blocks of size $b$ each), then  $\sum_{s=1}^rj_s=K$.
The whole signal range is divided into $n/L$ windows of size $L$, with $B,L$ satisfying $L\ge B$.
It is also assumed that any two consecutive nonzero clusters of $\bvec{x}$ are \emph{well-separated} by a zone of at least $L' = L + 2B - b$ zeros. Although, in principle, such a constraint is not necessary for the execution of the algorithm, it ensures that the range of indices, associated to an window identified by TSGBOMP, can contain only one nonzero cluster. This makes the analysis of the algorithm less complicated. 
Also, we assume that the signal length $n$ is large enough to accommodate any arrangement of nonzero clusters with $K$ blocks in the signal (of size $b$) such that any two consecutive clusters are separated by at least $L'$ zeros. 
%
%

The proposed TSGBOMP algorithm, given in Table~\ref{tab:TSGBOMP}, employs a two stage search procedure, of which the first one is similar to the BOMP. At any iteration $k (\ge 1)$ of the algorithm, it assumes that a residual vector $\bvec{r}^{k-1}$ and a partially constructed set $T^{k-1}$ are already available from step $k-1$ ($\bvec{r}^0 = \bvec{y}$, $T^0=\emptyset$). Then, following the BOMP procedure, it carries out a \emph{window-wise scanning} and
identifies a window of length $L$ and index $w^k$ from the range  $1$ to $n/L$, for which the $l_2$ norm of the correlation vector $\bvec{\Phi}[w^k]^H\bvec{r}^{k-1}$ is maximum. Next, it carries out a \emph{pointwise scanning} over the range $b^k = \{L(w^k - 1)-(B-1),\cdots, L w^k\}$ and identifies a cluster of size $B$ that has non-empty overlap with the chosen window and for which, the correlation vector $\bvec{\Phi}\curly{h^k}^H\bvec{r}^{k-1}$ has maximum $l_2$ norm, where $h^k$ denotes the set of the first indices of the elementary blocks (of size $b$) contained in the cluster, i.e., $h^k=\{i^k,\ i^k+b,\cdots,\ i^k+(p-1)b\}$, with $i^k \in b^k$. The set of indices $h^k$ is then appended to $T^{k-1}$ to construct $T^k$, and the residual vector is updated to $\bvec{r}^k$ by computing $\dualproj{T^k}\bvec{r}^k$.

%
\begin{table}[ht!]
	\centering
	\begin{tabular}{p{8cm}}
		\centering
		\hrulefill
		\begin{description}
			\item[\textbf{Input:}]\ Measurement vector $\bvec{y}\in \complex^m$,
			sensing matrix $\bvec{\Phi}\in \complex^{m\times n}$; sparsity level
			$K$; prespecified residual threshold $\epsilon$; window size ($L$); maximum block size $B=pb,\ p\ge 1$.
			\item[\textbf{Initialize:}]$\quad$
			Counter $k=0$, residue $\bvec{r}^0=\bvec{y}$, estimated support
			set, $T^0=\emptyset$.
			\item[\textbf{While}]($\norm{\bvec{r}^k}\ge
			\epsilon\ \mbox{and}\ \ k<K$)
			\begin{enumerate}
				\item $k=k+1$,
				\item $w^k = \argmax_{1\le l\le n/L} \norm{\bvec{\Phi}^H[l] \bvec{r}^{k-1}}$, $b^k = \{L(w^k - 1)+1-(B-1),\cdots, L w^k\} $,
				\item $i^k = \argmax_{i\in b^k}\norm{\bvec{\Phi}^H\curly{S}\bvec{r}^{k-1}}$, where $S = \{i,\ i+b,\cdots,\ i+(p-1)b\}$;\\
				$h^k = \{i^k,\ i^k+b,\cdots,\ i^k+(p-1)b\}$,
				\item $T^k = T^{k-1}\cup h^k$,
				\item $\displaystyle
				\bvec{x}^k\curly{T^k}=\argmin_{\bvec{u}}\norm{\bvec{y}-\bvec{\Phi}\curly{T^k}\bvec{u}}$,
				\item $\bvec{r}^k = \bvec{y}-\bvec{\Phi}\curly{T^k}\bvec{x}^k\curly{T^k}$.
			\end{enumerate}
			
			\item[\textbf{End While}]
		\end{description}
		\hrulefill
		\begin{description}
			\item[\textbf{Output:}]$\quad$  Estimated support set $I(\bvec{\Phi}\{T^K\})$ and estimated vector
			$\displaystyle
			 \widehat{\bvec{x}}=\argmin_{\bvec{u}:\supp(\bvec{u})=I(\bvec{\Phi}\{T^K\})}\norm{\bvec{y}-\bvec{\Phi}\curly{T^K}\bvec{u}}$.
		\end{description}
		\hrulefill
		\caption{Proposed TSGBOMP \textsc{Algorithm}}
		\label{tab:TSGBOMP}
	\end{tabular}
\end{table}
\section{Signal recovery using TSGBOMP}
\label{sec:signal-recovery-tsgbomp}
%

%
%
In this section, we derive sufficient recovery condition for TSGBOMP to successfully reconstruct an unknown vector $\bvec{x}\in \complex^n$ from a set of noisy measurements, given by $\bvec{y}(\in \complex^m)=\bvec{\Phi x}+\bvec{e}$, where $\bvec{e}$ is a measurement noise vector that is assumed to be $l_2$-bounded by a positive number $\epsilon$, i.e., $\norm{\bvec{e}}\le \epsilon$. For this, we follow
the approach of induction, i.e., at any $k$-th step of iteration ($k\ge 1$), we assume that in each of the previous $(k-1)$ steps, at least one true cluster of $\bvec{x}$ has been selected.
We then first find a condition ensuring that TSGBOMP identifies a window which has a non-empty intersection with one of the true (as yet unidentified) clusters of $\bvec{x}$. Then we find a condition that ensures that a correct cluster (of size between $b$ and $pb$) is identified from among all possible consecutive clusters having nonzero overlap with the identified window. Finally we combine these two conditions to find a uniform recovery condition under which TSGBOMP is successful in identifying a correct cluster at step $k$. For the analysis of finding a correct window, our proof closely follows the arguments of Theorem $1$ of~\cite{wen2018sharp} which is widely used for the analysis of BOMP using block sparse structure with predefined block boundaries and uniform block lengths. However, as the TSGBOMP does not assume fixed block boundaries and uniform block length, certain steps in the analysis of Wen~\emph{et al}~\cite{wen2018sharp} cannot be directly extended to TSGBOMP, and instead, certain novel structures are required to be defined to continue adopting the analysis of \cite{wen2018sharp}. These significantly change the final conditions for successful recovery.

%
\subsection{Condition for identifying a correct window at step $k (k\ge 1)$}
\label{sec:condition-success-window-selection}
To ensure success at iteration $k$, we need to ensure that the window selected in iteration $k$ has a nonempty overlap with at least one of the true clusters of $\bvec{x}$, having blocks with first indices in the set $T\setminus T^{k-1}$. To express this mathematically, let us define, for any subset $S\subseteq [n]$, $O_{S} = \left\{i:I\left(\bvec{\Phi}[i]\right)\cap I\left(\bvec{\Phi}\curly{S}\right)\ne \emptyset\right\},$ which is the set containing the first indices of the windows that have nonempty overlap with the blocks beginning with the indices in the set $S$.
Then the necessary condition for selecting a correct window at iteration $k$ is given by \begin{align}
\max_{i\in O_{T\setminus T^{k-1}}}\norm{\bvec{\Phi}^H[i] \bvec{r}^{k-1}} & > \max_{j\in O_{T\setminus T^{k-1}}^C}\norm{\bvec{\Phi}^H[j] \bvec{r}^{k-1}},\nonumber\\
\label{eq:step-k-window-selection-criterion}
\Leftrightarrow\opnorm{\bvec{\Phi}^H[O_{S^{k-1}}] \bvec{r}^{k-1}}{2,\infty} & > \opnorm{\bvec{\Phi}^H [O^C_{S^{k-1}}]\bvec{r}^{k-1}}{2,\infty},
\end{align}
where $S^{k-1} = T\setminus T^{k-1}$. To find a condition to ensure~\eqref{eq:step-k-window-selection-criterion}, we first observe that using steps similar to the ones used to derive Eq~($4.16$) of~\cite{wen2018sharp}, one obtains \begin{align}
\label{eq:rk-expression}
\bvec{r}^{k-1} & = \dualproj{T^{k-1}}\bvec{y}\nonumber\\
\ & = \dualproj{T^{k-1}}(\bvec{\Phi}\bvec{x}+\bvec{e})\nonumber\\
\ & = \dualproj{T^{k-1}}\left(\bvec{\Phi}\curly{S^{k-1}}\bvec{x}\curly{S^{k-1}}\right.\nonumber\\
\ &  \left.+\bvec{\Phi}\curly{T\cap T^{k-1}}\bvec{x}\curly{T\cap T^{k-1}} +\bvec{e}\right)\nonumber\\
\ & = \dualproj{T^{k-1}}\bvec{\Phi}\curly{S^{k-1}}\bvec{x}\curly{S^{k-1}} + \dualproj{T^{k-1}}\bvec{e}.
\end{align}
%
Plugging in the expression of $\bvec{r}^{k-1}$ from Eq.~\eqref{eq:rk-expression}, it is easy to see (using triangle and reverse triangle inequalities respectively) that the condition~\eqref{eq:step-k-window-selection-criterion} is satisfied if the following is ensured:\begin{align}
\lefteqn{\opnorm{\bvec{\Phi}^H[O_{S^{k-1}}]\dualproj{T^{k-1}}\bvec{\Phi}\curly{S^{k-1}}\bvec{x}\curly{S^{k-1}}}{2,\infty}} & &\nonumber\\
\ & - \opnorm{\bvec{\Phi}^H[O_{S^{k-1}}^C]\dualproj{T^{k-1}}\bvec{\Phi}\curly{S^{k-1}}\bvec{x}\curly{S^{k-1}}}{2,\infty} & \nonumber\\
\label{eq:step-k-window-selection-alternate-criterion}
> & \opnorm{\bvec{\Phi}^H[O_{S^{k-1}}]\dualproj{T^{k-1}}\bvec{e}}{2,\infty} + \opnorm{\bvec{\Phi}^H[O_{S^{k-1}}^C]\dualproj{T^{k-1}}\bvec{e}}{2,\infty}.
\end{align}
Now we will find a lower bound of the left hand side (LHS) and upper bound of the right hand side (RHS), of the inequality~\eqref{eq:step-k-window-selection-alternate-criterion} and compare them to come up with a sufficient condition to ensure~\eqref{eq:step-k-window-selection-alternate-criterion}.

In order to proceed further, we first note that one can write \begin{align}
\lefteqn{\opnorm{\bvec{\Phi}^H[O_{S^{k-1}}]\dualproj{T^{k-1}}\bvec{\Phi}\curly{S^{k-1}}\bvec{x}\curly{S^{k-1}}}{2,\infty}} & &\nonumber\\
\ & = \frac{\sum_{i\in O_{S^{k-1}}}\norm{\bvec{x}[i]}}{\sum_{i\in O_{S^{k-1}}}\norm{\bvec{x}[i]}}\cdot\nonumber\\
\ &  \opnorm{\bvec{\Phi}^H[O_{S^{k-1}}]\dualproj{T^{k-1}}\bvec{\Phi}\curly{S^{k-1}}\bvec{x}\curly{S^{k-1}}}{2,\infty}\nonumber\\
\ & \stackrel{(a)}{\ge} \frac{\sum_{i\in O_{S^{k-1}}}\abs{\bvec{x}^H[i]\bvec{\Phi}^H[i]\dualproj{T^{k-1}}\bvec{\Phi}\curly{S^{k-1}}\bvec{x}\curly{S^{k-1}}}}{\sum_{i\in O_{S^{k-1}}}\norm{\bvec{x}[i]}}\nonumber\\
\ & \stackrel{(b)}{\ge} \frac{\abs{\inprod{\bvec{\Phi}[O_{S^{k-1}}]\bvec{x}[O_{S^{k-1}}]}{\dualproj{T^{k-1}}\bvec{\Phi}\curly{S^{k-1}}\bvec{x}\curly{S^{k-1}}}}}{\sum_{i\in O_{S^{k-1}}}\norm{\bvec{x}[i]}}\nonumber\\
\ & \stackrel{(c)}{=}\frac{\abs{\inprod{\bvec{\Phi}\curly{S^{k-1}}\bvec{x}\curly{S^{k-1}}}{\dualproj{T^{k-1}}\bvec{\Phi}\curly{S^{k-1}}\bvec{x}\curly{S^{k-1}}}}}{\sum_{i\in O_{S^{k-1}}}\norm{\bvec{x}[i]}}\nonumber\\
\ & \stackrel{(d)}{\ge}\frac{\abs{\inprod{\bvec{\Phi}\curly{S^{k-1}}\bvec{x}\curly{S^{k-1}}}{\dualproj{T^{k-1}}\bvec{\Phi}\curly{S^{k-1}}\bvec{x}\curly{S^{k-1}}}}}{\sqrt{d_k}\norm{\bvec{x}\curly{S^{k-1}}}}\nonumber\\
\label{eq:step-k-window-selection-alternate-criterion-lhs-lower-bound-part1}
\ & = \frac{\norm{\dualproj{T^{k-1}}\bvec{\Phi}\curly{S^{k-1}}\bvec{x}\curly{S^{k-1}}}^2}{\sqrt{d_k}\norm{\bvec{x}\curly{S^{k-1}}}},
\end{align}
where $d_k = \abs{O_{S^{k-1}}}$. Here step $(a)$ uses the H\"older's-$(1,\infty)$ inequality followed by Cauchy-Schwartz inequality, while $(b)$ uses the triangle inequality. Step $(c)$ uses the observation that $I(\bvec{\Phi}\curly{S^{k-1}})\subset I(\bvec{\Phi}[O_{S^{k-1}}])$ along with the fact that a window can contain at most one nonzero block (of size $jb$, $1\le j\le p$), since such nonzero blocks are separated by at least $L$ zeros.
Finally, step $(d)$ uses Cauchy-Schwartz inequality
and the observations at step $(b)$ implying $\norm{\bvec{x}[O_{S^{k-1}}]} = \norm{\bvec{x}\curly{S^{k-1}}}$.

We now proceed to find an upper bound of $\opnorm{\bvec{\Phi}^H[O_{S^{k-1}}^C]\dualproj{T^{k-1}}\bvec{\Phi}\curly{S^{k-1}}\bvec{x}\curly{S^{k-1}}}{2,\infty}$. In order to do that, following the proof of Lemma~$4.1$ of~\cite{wen2018sharp}, we will instead find an upper bound of $\norm{\bvec{\Phi}^H[j]\dualproj{T^{k-1}}\bvec{\Phi}\curly{S^{k-1}}\bvec{x}\curly{S^{k-1}}}$ for an arbitrary $j\in O^C_{S^{k-1}}$ which will hold uniformly for all $j\in O^C_{S^{k-1}}$. 
Fix any $j\in O_{S^{k-1}}^C$ and let $\bvec{q}^{k-1}=\dualproj{T^{k-1}}\bvec{\Phi}\curly{S^{k-1}}\bvec{x}\curly{S^{k-1}}$. Following the proof of Lemma 4.1 of~\cite{wen2018sharp}, we now define a list of quantities for expressing $\norm{\bvec{\Phi}^H[j]\bvec{q}^{k-1}}$ in a convenient way. Let $\theta>0$ be an arbitrary positive number. Define,
 \begin{align}
\label{eq:mu-definition}
\mu & = -\frac{\sqrt{\theta + 1}-1}{\sqrt{\theta}},\\
\label{eq:h-definition}
h_l & = \frac{\left(\bvec{\Phi}[j]\right)_l^H \bvec{q}^{k-1}}{\norm{\bvec{\Phi}[j]^H\bvec{q}^{k-1}}},\ 1\le l\le L,\\
\label{eq:u-definition}
\ \bvec{u} & = \begin{bmatrix}
\bvec{x}\curly{S^{k-1}}\\
\bvec{0}
\end{bmatrix}\in \complex^{\abs{I\left(\bvec{\Phi}\curly{S^{k-1}}\right)} + L}, \\
\label{eq:w-definition}
\ \bvec{w} & = \mu\norm{\bvec{x}\curly{S^{k-1}}}\cdot\begin{bmatrix}
\bvec{0}\\
\bvec{h}\\
\end{bmatrix}\in \complex^{\abs{I\left(\bvec{\Phi}\curly{S^{k-1}}\right)} + L},\\
\label{eq:B-matrix-definition}
\bvec{B} & = \dualproj{T^{k-1}}\begin{bmatrix}
\bvec{\Phi}\curly{S^{k-1}} & \bvec{\Phi}[j]
\end{bmatrix}.
\end{align}
Using these definitions, it is straightforward to verify that $\bvec{B u} = \bvec{q}^{k-1}$ and $ \bvec{B w} = \mu\norm{\bvec{x}\curly{S^{k-1}}} \dualproj{T^{k-1}}\bvec{\Phi}[j]\bvec{h}$,  so that \begin{align*}
\lefteqn{\bvec{w}^H\bvec{B}^H \bvec{B u}} & & \\
\ & = \mu\norm{\bvec{x}\curly{S^{k-1}}} \bvec{h}^H \bvec{\Phi}^H[j](\dualproj{T^{k-1}})^H\bvec{q}^{k-1}\\
\ & =  \mu\norm{\bvec{x}\curly{S^{k-1}}} \bvec{h}^H \bvec{\Phi}^H[j]\bvec{q}^{k-1}\\
\ &  (\because (\dualproj{T^{k-1}})^H \dualproj{T^{k-1}} = \dualproj{T^{k-1}})\\
\ & = \mu\norm{\bvec{x}\curly{S^{k-1}}}\norm{\bvec{\Phi}[j]^H\bvec{q}^{k-1}}.
\end{align*} Consequently, it can be verified that the following identity holds (see the proof of Lemma~$4.1$ of~\cite{wen2018sharp} for details):\begin{align}
\lefteqn{\norm{\bvec{B}(\bvec{u}+\bvec{w})}^2 - \norm{\bvec{B}(\mu^2\bvec{u}-\bvec{w})}^2} & & \nonumber\\
\label{eq:step-k-window-selection-intermediate-identity1}
= & (1-\mu^4)\left(\norm{\bvec{B u}}^2 -\right.\nonumber\\ \ & \left.\sqrt{\theta}\norm{\bvec{x}\curly{S^{k-1}}}\norm{\bvec{\Phi}^H[j]\dualproj{T^{k-1}}\bvec{\Phi}\curly{S^{k-1}}\bvec{x}\curly{S^{k-1}}}\right).
\end{align}
Clearly, if a lower bound of the LHS of Eq.~\eqref{eq:step-k-window-selection-intermediate-identity1} can be found, one can find an upper bound of $\norm{\bvec{\Phi}^H[j]\dualproj{T^{k-1}}\bvec{\Phi}\curly{S^{k-1}}\bvec{x}\curly{S^{k-1}}}$. Consider the two terms of the LHS of Eq.~\eqref{eq:step-k-window-selection-intermediate-identity1}, i.e.,
$\norm{\bvec{B}(\bvec{u}+\bvec{w})}^2\equiv\norm{\dualproj{T^{k-1}}[\bvec{\Phi}\curly{S^{k-1}}\ \bvec{\Phi}[j]](\bvec{u}+\bvec{w})}^2$, and $\norm{\bvec{B}(\mu^2\bvec{u}-\bvec{w})}^2\equiv\norm{\dualproj{T^{k-1}}[\bvec{\Phi}\curly{S^{k-1}}\ \bvec{\Phi}[j]](\mu^2\bvec{u}-\bvec{w})}^2$. In the case of BOMP, both the window and block are the same entity, of length, say, $L$, meaning, $\bvec{\Phi}\curly{S^{k-1}}$ has an integral multiple of $L$ (i.e., $\abs{S^{k-1}}L$) number of columns, while $\bvec{\Phi}[j]$ has $L$ columns. Together,
$\bvec{\Phi}\curly{S^{k-1}}$ along with $\bvec{\Phi}[j]$ constitute the columns of $\bvec{\Phi}$ corresponding to the support of a conventional block sparse vector of block sparsity $\abs{S^{k-1}}+1$ and block size $L$. As a result, in the analysis of BOMP by Wen~\emph{et al}~\cite{wen2018sharp}, the properties of block RIP (See Sections 1 and 2 of~\cite{wen2018sharp}) could be leveraged to find an upper bound of the LHS of Eq.~\eqref{eq:step-k-window-selection-intermediate-identity1}. However, in the case of TSGBOMP, the columns of $\bvec{\Phi}\curly{S^{k-1}}$ correspond to the indices of the as yet unidentified true non-zero blocks of $\bvec{x}$ (totalling $\abs{S^{k-1}}b$ columns) and the columns of $\bvec{\Phi}[j]$ correspond to the columns of the window with index $j$ and size $L$. Since the unknown vector $\bvec{x}$ has a special non-uniform block structure with \textit{unspecified boundaries} (as described in Sec.~\ref{sec:proposed-algorithm}), together, these columns do not correspond to a conventional block sparse vector.  This necessitates definition of a new block sparse structure as given below.
\subsubsection{The Pseudoblock Interleaved Block Sparse Structure (PIBS)}
\label{sec:pseudoblock-interleaved-block-sparse-structure}
The proposed PIBS structure is a more general form of support set than given by the columns of the matrices $\bvec{\Phi}\curly{S^{k-1}}$ and $\bvec{\Phi}[j]$.
In this, clusters of consecutive non-zero blocks (called \emph{true} clusters from hereafter) with sizes given by integer multiples of a constant $b$, and having unspecified boundaries (analogous to $S^{k-1}$) are well separated by several indices, and these gaps may contain a second type of non-zero blocks that we
call \emph{pseudo} blocks (analogous to the  window $[j]$ in $\bvec{\Phi}[j]$). However, unlike above that has only one pseudo block $\bvec{\Phi}[j]$, we consider the more general case of $r$ pseudo blocks, $0\le r\le R$. Further, we do not constrain each pseudo block to remain confined to any specific window. Instead, they can be positioned anywhere in the cluster space such that they do not overlap with any of the true clusters. Positions not occupied by either the true clusters or the pseudo blocks are filled with zeros. The specific structure of a PIBS vector of length $n$ 
is described by a $6-$tuple $(b,p,l,L',K,R)$ that is explained below and is also illustrated in Fig~\ref{fig:PIBS-structure-illustration}:
\begin{figure}[t!]
	\centering
	\includegraphics[width=3in,height=1.5in]{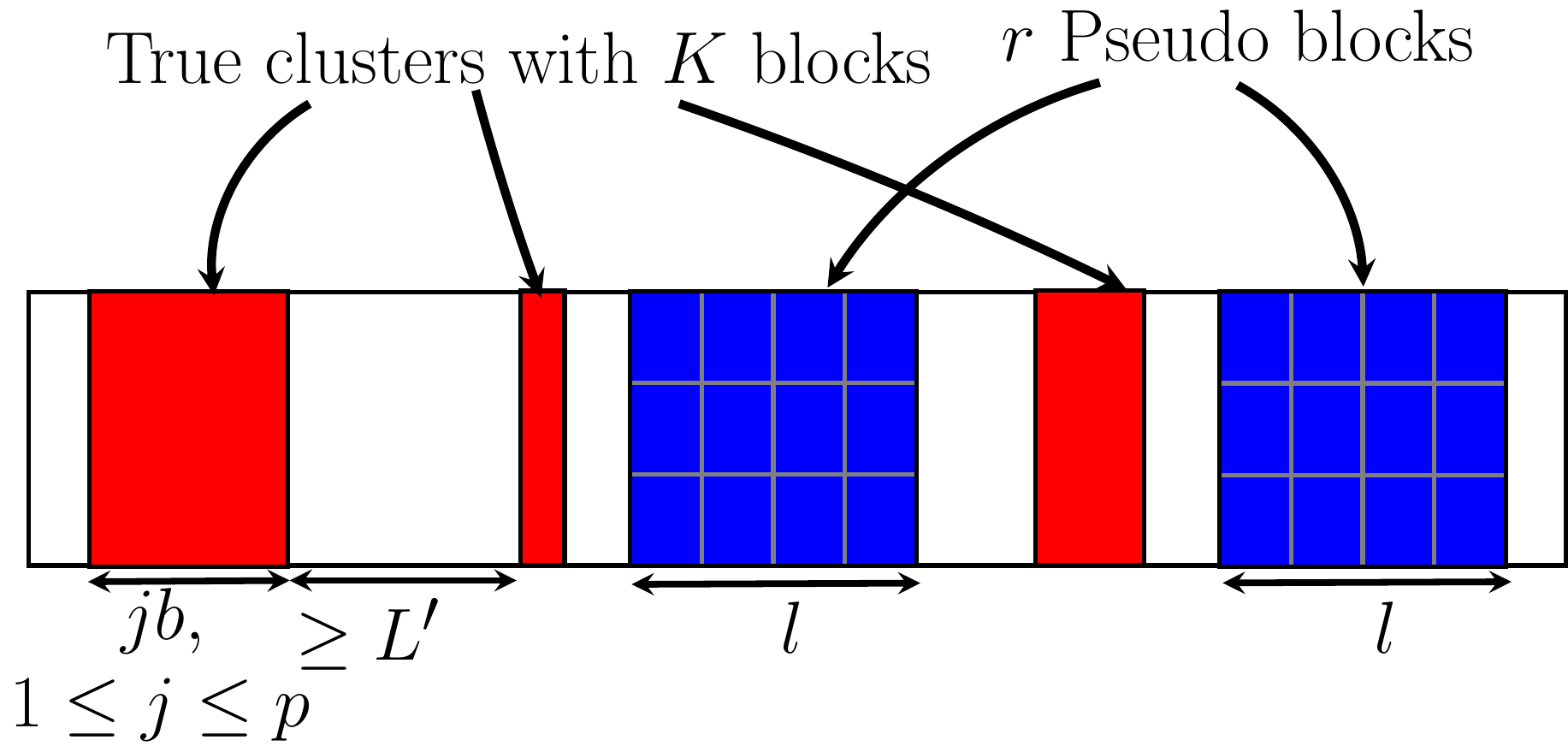}
	\caption{Illustration of a PIBS vector: there are $k$ ($\le K$) true clusters (shown in solid), each of size of the form $jb$, where $1\le j\le p$, and any two consecutive true clusters are separated by at least $L'$ zeros; these gaps (as well as the gaps between the signal edges and the first and last true clusters) contain $r (\le R)$ another type of nonzero blocks (shown in grids), called pseudo blocks, which have fixed size $l\ (0\le l\le L')$.}
	\label{fig:PIBS-structure-illustration}
\end{figure}
\begin{enumerate}
	\item There are $k$ true clusters, with $0\le k\le K$, having lengths $j_1b,\cdots,\ j_kb$, where $1\le j_s\le p,\ 1\le s\le k$, such that $\sum_{s=1}^k j_s = K$.
	\item Any two consecutive true clusters are separated by at least $L'$ zeros. We will always assume that $L'=L+2bp-b$, where $L$ is the length of a window, defined in Section~\ref{sec:notations}.
	\item There are $r$ pseudo blocks, $0\le r\le R$, which are a second type of nonzero blocks, that are of length $l$ each, $0\le l\le L'$, and are positioned anywhere but in a way so that they do not overlap with any of the $k$ true clusters defined above.
\end{enumerate}
Moreover, the signal length $n$ is assumed to be sufficiently large so that all the $R$ pseudo blocks, as well as $0\le k\le K$ true clusters can be accommodated within the signal without any overlap between the true clusters and the pseudo blocks.

Using this definition of the PIBS structure it can be easily seen that the columns of $\bvec{\Phi}\curly{S^{k-1}}$, along with the columns of $\bvec{\Phi}[j]$ correspond to the support of a PIBS vector with parameters $(b,p,L,L',c_k,1)$ ($c_k=|S^{k-1}|$), where the pseudo block is given by the window corresponding to $\bvec{\Phi}[j]$, and the true clusters correspond to $\bvec{\Phi}\curly{S^{k-1}}$.

Next we define a restricted isometry property for this PIBS signal structure, and state and prove a few Lemmas associated to it which will be required in the subsequent analysis of the TSGBOMP algorithm.
\subsubsection{Useful definitions and lemmas related to the PIBS structure}
\label{sec:useful-definitions-lemmas}
We first define an analog of the celebrated restricted isometry property (RIP)~\cite{candes_decoding_2005} in the context of the PIBS vectors.
\begin{define}
	\label{def:RIP}
	The pseudoblock interleaved block restricted isometry constant (PIBRIC) of order $(K,R)$ with parameters $b,p,l,L'$ of a matrix $\bvec{\Phi}\in \complex^{m\times n}$ is defined as \begin{align}
	\lefteqn{\delta_{b,p,l,L'}(K,R)} & &\nonumber\\
	\label{eq:PIBS-rip-eigenvalue-connection}
	\ & = \max_{S\subset [n]:S\in \cup_{r=0}^{R}\cup_{k=0}^{K}\Sigma_{b,p,l,L'}(k,r)}\opnorm{\bvec{\Phi}_S^H\bvec{\Phi}_S - \bvec{I}_{|S|\times |S|}}{2\to 2},
	\end{align}
	where $\Sigma_{b,p,l,L'}(k,r)$ is the collection of all subsets $S\subset [n]$ such that a PIBS vector with parameters $(b,p,l,L',k,r)$ can be supported on $S$.
\end{define}
The following lemma shows that the above definition is equivalent to the following interpretation of PIBRIC which is frequently used in the definition of RIP in literature.

\begin{lem}
		\label{lem:rip-PIBS-eigenvalue-connection}
		A matrix $\bvec{\Phi}\in \complex^{m\times n}$ having the PIBRIC $\delta_{b,p,l,L'}(K,R)$ of order $(K,R)$ with parameters $b,p,l,L'$ satisfies the following inequality for every PIBS vector $\bvec{x}\in \complex^{n}$ with parameters $(b,p,l,L',k,r),\;0\le k\le K,\;0\le r \le R$:
        \begin{align}
		\label{eq:PIBS-rip-inequalities}
		(1-\delta)\norm{\bvec{x}}^2 & \le \norm{\bvec{\Phi x}}^2\le (1+\delta)\norm{\bvec{x}}^2,
		\end{align}
for all $\delta\ge \delta_{b,p,l,L'}(K,R)$.
\end{lem}
\begin{proof}
	The proof is supplied in Appendix~\ref{sec:appendix-proof-rip-eigenvalue-condition}.
\end{proof}
Conventionally, in the literature the RIC has been required to be bounded above by a constant strictly smaller than $1$ to ensure success of compressed sensing algorithms, which is why we also aim to derive such an upper bound on a PIBRIC of certain order to ensure success of the TSGBOMP algorithm. In the sequel, we adopt the convention that a matrix $\bvec{\Phi}\in \complex^{m\times n}$ is said to satisfy pseudo block interleaved restricted isometry property (PIBRIP) of order $(K,R)$ with parameters $b,p,l,L'$ if $\delta_{b,p,l,L'}(K,R)\in (0,1)$.

Note that the PIBS vector with parameters $(b,p,l,L',K,R)$ is at most $Kb + Rl$ sparse in the conventional sense. From definition~\ref{def:RIP}, the computation of $\delta_{b,p,l,L'}(K,R)$ requires one to search for the maximum eigenvalue of $\opnorm{\bvec{\Phi}_S^H\bvec{\Phi}_S - \bvec{I}_S}{2\to 2}$ only over $\cup_{r=0}^R\cup_{k=0}^K\Sigma_{b,p,l,L'}(k,r)$, whereas the computation of $\delta_{Kb+Rl}$ requires finding the maximum eigenvalue of a similar submatrix over \emph{all} subsets $S$ of size $\le (Kb + Rl)$, meaning, the PIBRIC $\delta_{b,p,l,L'}(K,R)$ is smaller than the conventional RIC $\delta_{Kb + Rl}$. Similar observations were also made for the block sparse vectors in~\cite{eldar2009robust}. Such smaller restricted isometry constants are inherent in signals with specialized signal structure, which belong to a class of general signal structures referred to as model sparse signals~\cite{baraniuk2010model}.

Now we state the following lemmas which will be useful for our analysis of the TSGBOMP algorithm. Similar lemmas have already appeared in the literature in the context of the restricted isometry properties for sparse~\cite{candes_decoding_2005,dai2009subspace,davenport2010analysis,wen2017novel,wen2017nearly} and block sparse vectors~\cite{wen2018sharp}. We modify these lemmas to be applicable to our specific PIBS structure.
\begin{lem}[Block number monotonicity]
	\label{lem:PIBRIP-K-R-monotonicity}
	Let the matrix $\bvec{\Phi}$ satisfy PIBRIP with parameters $(b,p,l,L',K_i,R_j),\ 1\le i,\ j\le 2$, and $K_1\le K_2$, $R_1\le R_2$. Then, $\delta_{b,p,l,L'}(K_1,R_1)\le \delta_{b,p,l,L'}(K_2,R_1)\le \delta_{b,p,l,L'}(K_2,R_2)$, and $\delta_{b,p,l,L'}(K_1,R_1)\le \delta_{b,p,l,L'}(K_1,R_2)\le\delta_{b,p,l,L'}(K_2,R_2)$.
\end{lem}
\begin{proof}
	The proof is provided in Appendix.~\ref{sec:appendix-proof-pibrip-k-r-monotonicity}.
\end{proof}
\begin{lem}[]
	\label{lem:PIBRIP-L-monotonicity}
	For fixed $b,p,K$, if $0\le l< L'$, then $\delta_{b,p,l,L'}(K,R)\le \delta_{b,p,L',L'}(K,R)$ for $R=0,1,2$.
\end{lem}
\begin{proof}
	The proof is provided in Appendix.~\ref{sec:appendix-proof-pibrip-l-monotonicity}.
\end{proof}
\begin{lem}[]
	\label{lem:PIBRIP-B-L-monotonicity}
	For any $K\ge 1$ and $L'=L+2bp-b> L\ge bp$, $\delta_{b,p,L',L'}(K,1)\le \delta_{b,p,L',L'}(K-1,2)$.
\end{lem}
\begin{proof}
	The proof is provided in Appendix.~\ref{sec:appendix-proof-pibrip-b-l-monotonicity}.
\end{proof}
\begin{lem}[Projected matrix PIBRIP]
	\label{lem:projected-matrix-pibrip}
	Let $\bvec{\Phi}\in \complex^{m\times n}$ be a matrix such that $\delta_{b,p,l,L'}(K,R)<1$. Let $S$ be the support of a PIBS vector of length $n$ with parameters $(b,p,l,L',k,r),\;0\le k\le K,\;0\le r\le R$, i.e., $S$ contains the indices of the true clusters as well as the pseudo blocks of such a vector. Let $S_1,\;S_2\subset S$ be such that $S=I(\bvec{\Phi}\{S_1\})\cup {S_2}$ ($S_1$ contains the starting indices of some of the true blocks in $S$), and let $\bvec{x}\in \complex^n$  such that $\supp(\bvec{x})= S_2$. Then, \begin{align}
	(1-\delta_{b,p,l,L'}(K,R))\norm{\bvec{x}}^2 & \le \norm{\dualproj{S_1}\bvec{\Phi}_{S_2}\bvec{x}_{S_2}}^2\nonumber\\
	\label{eq:projected-PIBRIP-inequalities}
	\ & \le (1+\delta_{b,p,l,L'}(K,R))\norm{\bvec{x}}^2.
	\end{align}
\end{lem}
\begin{proof}
	The proof is provided in Appendix.~\ref{sec:appendix-proof-projected-matrix-pibrip}.
\end{proof}
\begin{lem}
	\label{lem:projected-matrix-innerproduct-inequality}
	Let the matrix $\bvec{\Phi}$ satisfy $\delta_{b,p,l,L'}(K,R)<1$. Let $S_1,S_2,S_3\subset [n]$ be such that the sets $S_2,\;S_3$ are disjoint and $I(\bvec{\Phi}\curly{S_1})\cup S_2\cup S_3$ is the support of a PIBS vector with parameters $(b,p,l,L',k,r),\;0\le k\le K,\;0\le r\le R$. Then, for any two vectors $\bvec{u}, \bvec{v}$ that are supported on the sets $S_2, S_3$ respectively, \begin{align}
	\label{eq:projected-matrix-innnerprouct-inequality}
	\abs{\inprod{\dualproj{S_1}\bvec{\Phi u}}{\bvec{\Phi v}}} & \le \delta_{b,p,l,L'}(K,R)\norm{\bvec{u}}\norm{\bvec{v}}.
	\end{align}
\end{lem}
\begin{proof}
	The proof is provided in Appendix.~\ref{sec:appendix-proof-projected-matrix-innerproduct-inequality}.
\end{proof}
\begin{lem}
	\label{lem:projected-column-lower-bound-wielandt-PIBS}
	Given that the matrix $\bvec{\Phi}$ has columns with unit norm and that it satisfies $\delta_{b,p,1,L'}(K,1)<1$. Let $S\subset[n]$ be such that $I(\bvec{\Phi}\curly{S})$ is the support of the true clusters of a PIBS vector with parameters $(b,p,l,L',k,r),\;0\le k\le K,\;0\le r\le R$, and let $j$ be an index such that $j\notin I(\bvec{\Phi}\curly{S})$. Then, \begin{align}
	\label{eq:projected-column-lower-bound-wielandt-pibrip}
	\norm{\dualproj{S}\bvec{\phi}_j}\ge \sqrt{1-\delta^2_{b,p,1,L'}(K,1)}.
	\end{align}
\end{lem}
\begin{proof}
	The proof is provided in Appendix.~\ref{sec:appendix-proof-projected-column-lower-bound-wielandt-pibs}.
\end{proof}
\subsection{Recovery guarantee}
Equipped with the definition of PIBS structure in Sec.~\ref{sec:pseudoblock-interleaved-block-sparse-structure}, it can be easily seen that the unknown vector $\bvec{x}$, described in Sec.~\ref{sec:proposed-algorithm}, has the structure of a PIBS vector of with parameters $(b,p,L,L',K,0)$. We now state a sufficient condition that the measurement matrix $\bvec{\Phi}$ as well as the unknown vector $\bvec{x}$ need to satisfy so that the TSGBOMP algorithm can exactly recover the support of $\bvec{x}$ within $K$ iterations.
\begin{thm}
	\label{thm:recovery-guarantee-tsgbomp}
	\textbf{Real case:} If $\bvec{x}\in \real^n,\ \bvec{\Phi}\in \real^{m\times n}$ and if the measurement matrix $\bvec{\Phi}$ as well as the unknown vector $\bvec{x}$, with parameters $(b,p,L,L',K,0)$ satisfy the conditions: \begin{align}
	\label{eq:tsgbomp-delta-recovery-condition}
	\delta & < \frac{1}{\sqrt{2K+1}},\\
	x_{\min} & > \frac{x_{\max}\delta\sqrt{B'b}}{(1+\delta)}\left[\frac{K}{4B'}\left(1+\delta\right) + \sqrt{K + 1}+1\right]\nonumber\\
	\label{eq:tsgbomp-noisy-recovery-condition}
	\ & + \frac{\sqrt{2(1+B')(1+\delta)}\epsilon}{(1-\delta\sqrt{2K+1})},
	\end{align}
	where $B'=pb-b+1,\ \delta:=\delta_{b,p,L',L'}(K-1,2)$, $x_{\min}=\min\{\abs{x_j}:j\in\supp(\bvec{x})\},\ x_{\max}=\max\{\abs{x_j}:j\in\supp(\bvec{x})\}$, the TSGBOMP algorithm can recover the support of $\bvec{x}$ exactly within $K$ iterations.
	
	\textbf{Complex case:} If $\bvec{x}\in \complex^n,\ \bvec{\Phi}\in \complex^{m\times n}$ and if the measurement matrix $\bvec{\Phi}$ as well as the unknown vector $\bvec{x}$, with parameters $(b,p,L,L',K,0)$ satisfy the conditions: \begin{align}
	\label{eq:tsgbomp-delta-recovery-condition-complex}
	\delta & < \frac{1}{\sqrt{2K+1}},
	\end{align}
	\begin{align}
	\lefteqn{x_{\min}} & & \nonumber\\
	\ & > \frac{\left[\delta\sqrt{Kb}+\frac{\delta(1-\delta^2)\sqrt{B'b}}{(1+\delta)}\left[\frac{K}{4B'}\left(1+\delta\right) + \sqrt{K + 1}+1\right]\right]x_{\max}}{\sqrt{(1-\delta^2)^2+\delta^2}}\nonumber\\
	\label{eq:tsgbomp-noisy-recovery-condition-complex}
	\ & + \frac{\sqrt{2(1+B')(1+\delta)}\epsilon}{(1-\delta\sqrt{2K+1})},
	\end{align}
	where $\delta:=\delta_{b,p,L',L'}(K-1,2)$, $x_{\min}=\min\{\abs{x_j}:j\in\supp(\bvec{x})\},\ x_{\max}=\max\{\abs{x_j}:j\in\supp(\bvec{x})\}$, the TSGBOMP algorithm can recover the support of $\bvec{x}$ exactly within $K$ iterations.
\end{thm}
\begin{proof}
	The derivations are detailed in Sections~\ref{sec:condition-success-window-selection},~\ref{sec:back-to-finding-condition-correct-window-step-k},~\ref{sec:condition-correct-block-selection-step-k}, and~\ref{sec:condition-overall-success}.
\end{proof}
%
\subsection{Back to finding a condition for identifying a correct window at step $k (k\ge 1)$}
\label{sec:back-to-finding-condition-correct-window-step-k}
We now return to complete the analysis of Sec.~\ref{sec:condition-success-window-selection} and proceed toward finding the sufficient condition for success of TSGBOMP at any step as provided by Theorem~\ref{thm:recovery-guarantee-tsgbomp}. Throughout the analysis we maintain that all the PIBS vectors considered will have length $n$ where $n$ is chosen large enough to accommodate all possible configurations of the different PIBS vectors emerging in the analysis \footnote{It can be shown that by taking $n\ge Kb+2L'+K+1$ this property is satisfied.}. We will now find an upper bound of the left hand side (LHS) of Eq.~\eqref{eq:step-k-window-selection-intermediate-identity1} by finding the PIBRIC of the associated PIBRIP of $\bvec{B}=\dualproj{T^{k-1}}[\bvec{\Phi}\curly{S^{k-1}}\ \bvec{\Phi}[j]]$.

First, observe that it can be now easily verified that both the vectors $\bvec{u}+\bvec{w}$ as well as $\mu^2\bvec{u}-\bvec{w}$ have the common support (for $\mu\ne 0$) $S'=I(\bvec{\Phi}\curly{S^{k-1}})\cup I(\bvec{\Phi}[j])$ which corresponds to the support of $1$ pseudo block of length $L$ and $c_k(=\abs{S^{k-1}})$ true clusters of a PIBS vector with parameters $(b,p,L,L',c_k,1)$. Now, let us denote $\widetilde{S}=I(\bvec{\Phi}\curly{T^{k-1}})\cup S'$. Note that as $j\in O^C[S^{k-1}]$ the window $j$ does not have any overlap with $I(\bvec{\Phi}\curly{S^{k-1}})$, but might have an overlap with $I(\bvec{\Phi}\curly{T^{k-1}})$. Consider first the case that the window $j$ does not have any overlap with $I(\bvec{\Phi}\curly{T^{k-1}})$. Then $\widetilde{S}$ corresponds to the support of a PIBS vector with parameters $(b,p,L,L',K,1)$. Consequently, according to Lemma~\ref{lem:projected-matrix-pibrip}, the matrix $\bvec{B} =\dualproj{T^{k-1}}\bvec{\Phi}_{S'}$ satisfies PIBRIP with PIBRIC given by $\delta_{b,p,L,L'}(K,1)$, which is upper bounded by $\delta_{b,p,L',L'}(K,1)$ by Lemma~\ref{lem:PIBRIP-L-monotonicity}. On the other hand, if there is a nonempty overlap of the window $j$ with $I(\bvec{\Phi}\curly{T^{k-1}})$, there can be overlap with at most one of the true clusters in $I(\bvec{\Phi}\curly{T^{k-1}})$. Call that true cluster (i.e. the set of its indices) $\mathcal{C}'$. Now, there are two cases to consider. In one case if $\mathcal{C}'$ is a proper subset of the window $j$, the set $\widetilde{S}$ corresponds to the support of a PIBS vector with parameters $(b,p,L,L',K-k',1)$, where $k'$ is the number of true blocks in the cluster $\mathcal{C}'$. In this case, by Lemma~\ref{lem:projected-matrix-pibrip} the matrix $\bvec{B}$ satisfies PIBRIP with PIBRIC $\delta_{b,p,L,L'}(K-k',1)$, which is upper bounded by $\delta_{b,p,L',L'}(K,1)$ by Lemmas~\ref{lem:PIBRIP-K-R-monotonicity} and~\ref{lem:PIBRIP-L-monotonicity}. On the other hand, if $\mathcal{C}'$ has only partial overlap with the window $j$, with say $l'$ being the size of the overlap (clearly, $1\le l'< pb\le L$), the set $\widetilde{S}$ corresponds to the support of a PIBS vector with parameters $(b,p,L-l',L',K,1)$. In this case, by Lemma~\ref{lem:projected-matrix-pibrip} the matrix $\bvec{B}$ satisfies PIBRIP with PIBRIC given by $\delta_{b,p,L-l',L'}(K,1)$, which in turn is upper bounded by $\delta_{b,p,L',L'}(K,1)$ by Lemma~\ref{lem:PIBRIP-L-monotonicity}. Therefore, using the expressions of $\bvec{u},\bvec{w}$ from Eqs.~\eqref{eq:u-definition} and~\eqref{eq:w-definition} and the fact that $\bvec{u}$ and $\bvec{w}$ are orthogonal, we obtain, \begin{align}
\lefteqn{\norm{\bvec{B}(\bvec{u}+\bvec{w})}^2 - \norm{\bvec{B}(\mu^2\bvec{u}-\bvec{w})}^2} & & \nonumber\\
\ & \ge (1-\delta_{b,p,L',L'}(K,1))\norm{\bvec{u}+\bvec{w}}^2\nonumber\\
\ & -(1+\delta_{b,p,L',L'}(K,1))\norm{\mu^2\bvec{u}-\bvec{w}}^2\nonumber\\
\ & = \left[(1-\delta_{b,p,L',L'}(K,1))-\mu^4(1+\delta_{b,p,L',L'}(K,1))\right]\norm{\bvec{u}}^2\nonumber\\
\ & -2\delta_{b,p,L',L'}(K,1) \norm{\bvec{w}}^2,\nonumber\\
\ & = (1-\mu^4)\left[1-\frac{1+\mu^4}{1-\mu^4}\delta_{b,p,L',L'}(K,1)\right]\norm{\bvec{x}\curly{S^{k-1}}}^2\nonumber\\
\ & -2\mu^2\delta_{b,p,L',L'}(K,1)\norm{\bvec{x}\curly{S^{k-1}}}^2\nonumber\\
\ & = (1-\mu^4)\left[1 - \frac{1+\mu^2}{1-\mu^2}\delta_{b,p,L',L'}(K,1)\right]\norm{\bvec{x}\curly{S^{k-1}}}^2.
\end{align}
Now, using the expression of $\mu$ from Eq.~\eqref{eq:mu-definition}, we obtain, \begin{align}
\frac{1+\mu^2}{1-\mu^2} & = \frac{\theta+\left(\sqrt{\theta+1}-1\right)^2}{\theta-\left(\sqrt{\theta+1}-1\right)^2}\nonumber\\
\ & =\frac{2\theta+2-2\sqrt{\theta+1}}{2\sqrt{\theta+1}-2}=\sqrt{\theta+1}.
\end{align} Therefore, one can find a lower bound of the RHS of Eq.~\eqref{eq:step-k-window-selection-intermediate-identity1} as the following: \begin{align}
\lefteqn{\norm{\bvec{B}(\bvec{u}+\bvec{w})}^2 - \norm{\bvec{B}(\mu^2\bvec{u}-\bvec{w})}^2} & & \nonumber\\
\label{eq:step-k-window-selection-intermediate-inequality1}
\ge & (1-\mu^4)\norm{\bvec{x}\curly{S^{k-1}}}^2\left(1-\sqrt{\theta + 1}\delta_{b,p,L',L'}(K,1)\right).
\end{align}
Since, it is easy to verify that $\abs{\mu} < 1$, and since the identity~\eqref{eq:step-k-window-selection-intermediate-identity1} and the inequality~\eqref{eq:step-k-window-selection-intermediate-inequality1} hold true for all $j\in O_{S^k}^C$, it follows that \begin{align}
\lefteqn{\norm{\bvec{B u}}^2 -} & &\nonumber\\
\ & \sqrt{\theta}\norm{\bvec{x}\curly{S^{k-1}}}\nonumber\\
\ & \cdot\opnorm{\bvec{\Phi}^H[O_{S^{k-1}}^C]\dualproj{T^{k-1}}\bvec{\Phi}\curly{S^{k-1}}\bvec{x}\curly{S^{k-1}}}{2,\infty}\nonumber\\
\label{eq:step-k-window-selection-alternate-criterion-rhs-upper-bound-part2}
\ge & \norm{\bvec{x}\curly{S^{k-1}}}^2\left(1-\sqrt{\theta + 1}\delta_{b,p,L',L'}(K,1)\right).
\end{align}
Therefore, using~\eqref{eq:step-k-window-selection-alternate-criterion-lhs-lower-bound-part1},~\eqref{eq:step-k-window-selection-alternate-criterion-rhs-upper-bound-part2}, and $\bvec{B u} = \dualproj{T^{k-1}}\bvec{\Phi}\curly{S^{k-1}}\bvec{x}\curly{S^{k-1}}$, it follows that \begin{align}
\lefteqn{\opnorm{\bvec{\Phi}^H[O_{S^{k-1}}]\dualproj{T^{k-1}}\bvec{\Phi}\curly{T\setminus T^{k-1}}\bvec{x}\curly{S^{k-1}}}{2,\infty}} & &\nonumber\\
\ & - \sqrt{\frac{\theta}{d_k}}\opnorm{\bvec{\Phi}^H[O_{S^{k-1}}^C]\dualproj{T^{k-1}}\bvec{\Phi}\curly{S^{k-1}}\bvec{x}\curly{S^{k-1}}}{2,\infty} & &\nonumber\\
\label{eq:step-k-window-selection-alternate-criterion-rhs-upper-bound-noiseless-prelim}
\ge & \frac{\left(1-\sqrt{\theta + 1}\delta_{b,p,L',L'}(K,1)\right)\norm{\bvec{x}\curly{S^{k-1}}}}{\sqrt{d_k}}.
\end{align}
Clearly, putting $\theta=d_k$ we recover the LHS of Eq.~\eqref{eq:step-k-window-selection-alternate-criterion} from the LHS of Eq.~\eqref{eq:step-k-window-selection-alternate-criterion-rhs-upper-bound-noiseless-prelim}. Consequently, we obtain, \begin{align}
\lefteqn{\opnorm{\bvec{\Phi}^H[O_{S^{k-1}}]\dualproj{T^{k-1}}\bvec{\Phi}\curly{T\setminus T^{k-1}}\bvec{x}\curly{S^{k-1}}}{2,\infty}} & &\nonumber\\
\ & - \opnorm{\bvec{\Phi}^H[O_{S^{k-1}}^C]\dualproj{T^{k-1}}\bvec{\Phi}\curly{S^{k-1}}\bvec{x}\curly{S^{k-1}}}{2,\infty} & &\nonumber\\
\label{eq:step-k-window-selection-alternate-criterion-rhs-upper-bound-noiseless}
\ge & \frac{\left(1-\sqrt{d_k + 1}\delta_{b,p,L',L'}(K,1)\right)\norm{\bvec{x}\curly{S^{k-1}}}}{\sqrt{d_k}}.
\end{align} 
To find an upper bound on the RHS of inequality~\eqref{eq:step-k-window-selection-alternate-criterion}, we follow the derivation of the inequality~$(4.22)$ in~\cite{wen2018sharp}. First we note that there exists windows indexed by $i_0\in O_{S^{k-1}}$, and $j_0\in O^C_{S^{k-1}}$, such that $\opnorm{\bvec{\Phi}^H[O_{S^{k-1}}]\dualproj{T^{k-1}}\bvec{e}}{2,\infty}=\norm{\bvec{\Phi}^H[i_0]\dualproj{T^{k-1}}\bvec{e}}$, and $\opnorm{\bvec{\Phi}^H[O^C_{S^{k-1}}]\dualproj{T^{k-1}}\bvec{e}}{2,\infty}=\norm{\bvec{\Phi}^H[j_0]\dualproj{T^{k-1}}\bvec{e}}$. Therefore, \begin{align}
\ & \opnorm{\bvec{\Phi}^H[O_{S^{k-1}}]\dualproj{T^{k-1}}\bvec{e}}{2,\infty} + \opnorm{\bvec{\Phi}^H[O_{S^{k-1}}^C]\dualproj{T^{k-1}}\bvec{e}}{2,\infty} \nonumber\\
\ & = \norm{\bvec{\Phi}^H[i_0]\dualproj{T^{k-1}}\bvec{e}} + \norm{\bvec{\Phi}^H[j_0]\dualproj{T^{k-1}}\bvec{e}}\nonumber\\
\ & \le \sqrt{2}\norm{\bvec{\Phi}^H[i_0\cup j_0]\dualproj{T^{k-1}}\bvec{e}}\nonumber\\
\ & \le \sqrt{2}\sigma_{\max}\left(\bvec{\Phi}^H[i_0\cup j_0]\dualproj{T^{k-1}}\right)\norm{\bvec{e}}\nonumber\\
\ & = \sqrt{2\lambda_{\max}\left(\bvec{\Phi}^H[i_0\cup j_0]\dualproj{T^{k-1}}\bvec{\Phi}[i_0\cup j_0]\right)}\norm{\bvec{e}},
\end{align}
where in the last three steps, for any matrix $\bvec{A}$, $\sigma_{\max}(\bvec{A})$ denotes the maximum singular value of $\bvec{A}$ and for any Hermitian matrix $\bvec{B}$, $\lambda_{\max}(\bvec{B})$ denotes the largest eigenvalue of $\bvec{B}$.

Now, observe that although the window $i_0$ cannot have an overlap with the set $I(\bvec{\Phi}\curly{T^{k-1}})$ , there might be overlap of the later with window $j_0$. If there is no overlap, then the set $S''=I(\bvec{\Phi}\curly{T^{k-1}})\cup I(\bvec{\Phi}[i_0\cup j_0])$ corresponds to the support of a PIBS vector with parameters $(b,p,L,L',\abs{T^{k-1}},2)$. If there is overlap, let $\mathcal{C}''$ be the true cluster in $I(\bvec{\Phi}\curly{T^{k-1}})$ which has a nonempty overlap with the window $j_0$. If $\mathcal{C}''$ is a subset of window $j_0$, the set $S''$ corresponds to the support of a PIBS vector with parameters $(b,p,L,L',\abs{T^{k-1}}-k'',2)$, where $k''$ is the number of true blocks in the cluster $\mathcal{C}''$ (note that $1\le k''\le \min\{p,\abs{T^{k-1}}\}$). On the other hand, if 
$\mathcal{C}''$ has only partial overlap with window $j_0$, then, assuming that the length of overlap between $\mathcal{C}''$ and the window $j_0$ is $l''(1\le l''\le k''b-1)$, the set $S''$ can be covered by another set $S'''$ which consists of the unions of the windows $i_0,j_0$, the set $I(\bvec{\Phi}\curly{T^{k-1}}) \setminus \mathcal{C}''$ and another true cluster of size $k''b$ obtained
by prefixing or suffixing (as the case may be) $l''$ indices to the non-overlapping side of $\mathcal{C}''\setminus [j_0]$ [$[j_0]$ indicates the set of indices covered by the window $j_0$].
Clearly, The set $S'''$ corresponds to the support of a PIBS vector with parameters $(b,p,L,L',\abs{T^{k-1}},2)$.

Therefore, usng the Lemmas~\ref{lem:projected-matrix-pibrip},~\ref{lem:PIBRIP-K-R-monotonicity} as well as~\ref{lem:PIBRIP-L-monotonicity}, it can be seen that the matrix $\dualproj{T^{k-1}}\bvec{\Phi}[i_0\cup j_0]$ satisfies PIBRIP with PIBRIC given by $\delta_{b,p,L',L'}(\abs{T^{k-1}},2)$. Consequently, $\lambda_{\max}\left(\dualproj{T^{k-1}}\bvec{\Phi}[i_0\cup j_0]\right)\le (1+\delta_{b,p,L',L'}(\abs{T^{k-1}},2))$. Therefore,
\begin{align}
\ & \opnorm{\bvec{\Phi}^H[O_{S^{k-1}}]\dualproj{T^{k-1}}\bvec{e}}{2,\infty} + \opnorm{\bvec{\Phi}^H[O^C_{S^{k-1}}]\dualproj{T^{k-1}}\bvec{e}}{2,\infty}\nonumber\\
\label{eq:step-k-window-selection-alternate-criterion-rhs-noise-part-upper-bound}
\ & \le \sqrt{2(1+\delta_{b,p,L',L'}(\abs{T^{k-1}},2))}\epsilon,
\end{align}
where we have assumed that $\norm{\bvec{e}}\le \epsilon$.
Thus, using inequalities~\eqref{eq:step-k-window-selection-alternate-criterion},~\eqref{eq:step-k-window-selection-alternate-criterion-rhs-upper-bound-noiseless} and~\eqref{eq:step-k-window-selection-alternate-criterion-rhs-noise-part-upper-bound} the following sufficient condition is derived to ensure a correct window selection at step $k (k\ge 1)$:
\begin{align}
\lefteqn{\frac{\left(1-\sqrt{d_k + 1}\delta_{b,p,L',L'}(K,1)\right)\norm{\bvec{x}\curly{S^{k-1}}}}{\sqrt{d_k}}} & &\nonumber\\
\label{eq:step-k-sufficient-condition-correct-window-selection}
> & \sqrt{2(1+\delta_{b,p,L',L'}(\abs{T^{k-1}},2))}\epsilon.
\end{align}
\subsection{Condition for true cluster selection at step $k (k\ge 1)$}
\label{sec:condition-correct-block-selection-step-k}
We assume that a correct window, indexed by $w^k$ (that is the $w^k$th window), has already been selected at step $k$, i.e. set of the columns of $\bvec{\Phi}[w^k]$ has a nonempty intersection with the set of columns of $\bvec{\Phi}\curly{S^{k-1}}$. We now find a condition to ensure that a true cluster from $I(\bvec{\Phi}\curly{S^{k-1}})$, having a nonempty overlap with the window indexed $w^k$, is selected at step $k$.

Let the set of indices corresponding to the true cluster having a non-empty overlap with window $w^k$ be denoted by $C^k$ and let $t\ (1\le t \le p)$ be the number of true blocks in $C^k$. Clearly, $C^k\subset \beta^k$, where $\beta^k$ is the set of indices $L(w^k-1)+1-(B-1),L(w^k-1)-(B-1)+2,\cdots,Lw^k+B - 1$. Note that the assumption that any two consecutive true clusters are separated by at least $L'=L+2bp-b$ zeros, ensures that the set $\beta^k$ does not have \emph{any} overlap with with a true cluster from $I(\bvec{\Phi}\curly{S^{k-1}})$ other than $C^k$, that is, it ensures that $I(\bvec{\Phi}\curly{S^{k-1}})\cap \beta^k\setminus C^k=\emptyset$.

Let $\mathcal{W}$ be the collection of all sets $S$ of size $B$ such that $S\subset \beta^k$ and such that $S$ covers $C^k$, i.e., $C_k\subseteq S$. On the other hand, let $\mathcal{W}'$ be the collection of all sets $S'$ of size $B$ such that $S'\subset \beta^k$ and that $S'$ does not cover $C^k$, i.e., $C^k\not\subseteq S'$. Then, a set $S\in \mathcal{W}$ is selected at step $k(\ge 1)$ if and only if \begin{align}
\label{eq:step-k-block-selection-criterion}
\max_{S\in \mathcal{W}} \norm{\bvec{\Phi}^H_S\bvec{r}^{k-1}} > \max_{S'\in \mathcal{W}'} \norm{\bvec{\Phi}^H_{S'}\bvec{r}^{k-1}}.
\end{align}
Now, let $S_0=\argmax_{S\in \mathcal{W}} \norm{\bvec{\Phi}^H_S\bvec{r}^{k-1}}$. Consider any $S'\in \mathcal{W}'$. Let $s'=\abs{S'\cap C^k},\ t'=\abs{S'\cap S_0\setminus C^k}$. Note that \begin{align}
\norm{\bvec{\Phi}^H_{S_0\setminus S'}\bvec{r}^{k-1}} & \ge \norm{\bvec{\Phi}^H_{C_k\setminus S'}\bvec{r}^{k-1}}\nonumber\\
\ & \ge \sqrt{tb-s'}\min_{j\in C^k}\abs{\inprod{\bvec{\phi}_j}{\bvec{r}^{k-1}}},
\end{align}
where we have used $\abs{C_k\setminus S'}=\abs{C_k}-\abs{C_k\cap S'}=tb-s'$. On the other hand, \begin{align}
\norm{\bvec{\Phi}^H_{S'\setminus S_0}\bvec{r}^{k-1}} & \le \sqrt{pb-s'-t'} \max_{l\in \beta^k\setminus C^k}\abs{\inprod{\bvec{\phi}_l}{\bvec{r}^{k-1}}},
\end{align}
where we have used $\abs{S'\setminus S_0}=\abs{S'}-\abs{S'\cap S_0}=\abs{S'}-\abs{S'\cap S_0\setminus C^k}-\abs{S'\cap C^k}=pb-t'-s'$. Therefore, the inequality~\eqref{eq:step-k-block-selection-criterion} is satisfied if \begin{align}
\label{eq:step-k-block-selection-criterion-preilim}
\min_{j\in C^k}\abs{\inprod{\bvec{\phi}_j}{\bvec{r}^{k-1}}} & \ge \sqrt{\frac{pb-t'-s'}{tb-s'}} \max_{l\in \beta^k\setminus C^k}\abs{\inprod{\bvec{\phi}_l}{\bvec{r}^{k-1}}},
\end{align}
where in the above we have used the fact that $0\le s'\le tb-1$ since, by definition, $S'$ cannot fully cover $C^k$.

Now, note that $\frac{pb-t'-s'}{tb-s'}=\frac{(p-t)b-t'}{tb-s'}+1$ is an increasing function of $s'$ since $0\le t'=\abs{S'\cap S_0\setminus C^k}\le\abs{S_0\setminus C^k}=(p-t)b$. Since $s'\le tb-1$, we obtain that $\frac{pb-t'-s'}{tb-s'}\le (p-t)b-t'+1\le pb-b+1$, where we have used the facts $t\ge 1,\ t'\ge 0$. Therefore, the inequality~\eqref{eq:step-k-block-selection-criterion-preilim} is satisfied if
the following is satisfied\footnote[2]{Similar sufficient condition was derived in~\cite{kannu2018spcom} in the context of a coherence based analysis of the sliding-block type algorithm proposed therein.}: \begin{align}
\label{eq:step-k-block-selection-alternative-criterion}
\min_{j\in C^k}\abs{\inprod{\bvec{\phi}_j}{\bvec{r}^{k-1}}} & \ge \sqrt{B'} \max_{l\in \beta^k\setminus C^k}\abs{\inprod{\bvec{\phi}_l}{\bvec{r}^{k-1}}},
\end{align}
where $B'=pb-b+1$.

Fix any $j_1\in C^k,\ j_2\in \beta^k\setminus C^k$. Now, we obtain
\begin{align}
\lefteqn{\abs{\inprod{\bvec{\phi}_{j_1}}{\bvec{r}^{k-1}}} - \sqrt{B'}\abs{\inprod{\bvec{\phi}_{j_2}}{\bvec{r}^{k-1}}}} & &\nonumber\\
\stackrel{(d)}{\ge } & \underbrace{\abs{\inprod{\bvec{\phi}_{j_1}}{\dualproj{T^{k-1}}\bvec{\Phi}\curly{S^{k-1}}\bvec{x}\curly{S^{k-1}}}}}_{F_1}\nonumber\\
\ & - \sqrt{B'}\underbrace{\abs{\inprod{\bvec{\phi}_{j_2}}{\dualproj{T^{k-1}}\bvec{\Phi}\curly{S^{k-1}}\bvec{x}\curly{S^{k-1}}}}}_{F_2}\nonumber\\
\label{eq:step-k-block-selection-alternative-criterion-lhs-preliminary-lower-bound}
\ & - \underbrace{\left(\abs{\inprod{\bvec{\phi}_{j_1}}{\dualproj{T^{k-1}}\bvec{e}}} + \sqrt{B'}\abs{\inprod{\bvec{\phi}_{j_2}}{\dualproj{T^{k-1}}\bvec{e}}}\right)}_{F_3}.
\end{align}
Here, step $(d)$ uses the expression for $\bvec{r}^{k-1}$ from Eq.~\eqref{eq:rk-expression} and the reverse triangle inequality and triangle inequality, respectively. 
 We now proceed to find upper bounds of $F_2,F_3$ and a lower bound of $F_1$.

First consider $F_2$. To find its upper bound, a procedure exactly similar to the one used via ~\eqref{eq:step-k-window-selection-intermediate-identity1},~\eqref{eq:step-k-window-selection-intermediate-inequality1} and~\eqref{eq:step-k-window-selection-alternate-criterion-rhs-upper-bound-part2} to calculate an upper bound of a similar quantity $\norm{\bvec{\Phi}^H[j]\dualproj{T^{k-1}}\bvec{\Phi}\curly{S^{k-1}}\bvec{x}\curly{S^{k-1}}}$ can be used. Since, like the window $[j]$ which is disjoint to $I(\bvec{\Phi}\curly{S^{k-1}})$, the column $j_2\notin I(\bvec{\Phi}\curly{S^{k-1}})$, this will imply simply replacing the window $[j]$ of length $L$ by a window of length 1 consisting of the column $j_2$ only. 
We make corresponding changes in the definitions of $\mu$, $\bvec{h},\ \bvec{u},\ \bvec{w},\ \bvec{B}$ as given by ~\eqref{eq:mu-definition},~\eqref{eq:h-definition},~\eqref{eq:u-definition},~\eqref{eq:w-definition} and~\eqref{eq:B-matrix-definition} respectively, by replacing $\theta$ by some positive number $\alpha$, $L$ by $1$, and $\bvec{\Phi}[j]$ by $\bvec{\phi}_{j_2}$.
To describe the structure of the PIBS vector that emerges as a result of a similar analysis, note that since $j_2\in \beta^k\setminus C^k$, we always have $j_2 \notin I(\bvec{\Phi}\curly{T^{k-1}})$.
As a result, we have a PIBS vector with
parameters $(b,p,1,L',K,1)$ (corresponding PIBRIC : $\delta_{b,p,1,L'}(K,1)$). 
Consequently, following the steps of ~\eqref{eq:step-k-window-selection-intermediate-identity1},~\eqref{eq:step-k-window-selection-intermediate-inequality1} and~\eqref{eq:step-k-window-selection-alternate-criterion-rhs-upper-bound-part2}, we obtain the following inequality (for an arbitrary positive number $\alpha$):
\begin{align}
\lefteqn{\norm{\dualproj{T^{k-1}}\bvec{\Phi}\curly{S^{k-1}}\bvec{x}\curly{S^{k-1}}}^2 - \sqrt{\alpha}\norm{\bvec{x}\curly{S^{k-1}}} F_2} & &\nonumber\\
\label{eq:F2_upper_bound}
\ & \ge \norm{\bvec{x}\curly{S^{k-1}}}^2\left(1-\sqrt{\alpha + 1}\delta_{b,p,1,L'}(K,1)\right).
\end{align}
To find an upper bound of $F_3$, we derive:
\begin{align}
F_3 & \le \sqrt{1+B'}\norm{\bvec{\Phi}_{U}^H\dualproj{T^{k-1}}\bvec{e}}\nonumber\\
\ & \le \sqrt{1+B'}\sqrt{\lambda_{\max}\left(\bvec{\Phi}_{U}^H\dualproj{T^{k-1}}\bvec{\Phi}_{U}\right)}\norm{\bvec{e}} \nonumber\\
\label{eq:F3-upper-bound}
\ & \le \sqrt{(B'+1)(1+\delta_{b,p,1,L'}(\abs{T^{k-1}},2))}\epsilon,
\end{align}
where $U$ is the set of the indices $j_1,j_2$, and the last step follows from Lemma~\ref{lem:projected-matrix-pibrip} which uses the observation that $I(\bvec{\Phi}\curly{T^{k-1}})\cup U$ corresponds to the support of a PIBS vector with parameters $(b,p,1,L',\abs{T^{k-1}},2)$ (which is true because $j_1,j_2\notin I(\bvec{\Phi}\curly{T^{k-1}})$ and $j_1\ne j_2$).

Now, we proceed to find a lower bound of $F_1$. For this, first we define, for any $z\in \complex,z\ne 0,\ \sign{z}=\frac{z}{\abs{z}}$ (i.e., if $z=re^{j\theta}$, then $\sign{z}=e^{j\theta}$).
Then recalling that $\bvec{q}^{k-1}=\dualproj{T^{k-1}}\bvec{\Phi}\curly{S^{k-1}}\bvec{x}\curly{S^{k-1}}$, we have $F_1=\abs{\inprod{\bvec{\phi}_{j_1}}{\bvec{q}^{k-1}}}=\inprod{s_{j_1}^*\bvec{\phi}_{j_1}}{\bvec{q}^{k-1}}\equiv \inprod{s_{j_1}^*\dualproj{T^{k-1}}\bvec{\phi}_{j_1}}{\bvec{q}^{k-1}}$, 
where, $s_{j_1}=\sign{\inprod{\bvec{\phi}_{j_1}}{\bvec{q}^{k-1}}}$. Also, for two vectors ${\bvec x},\;{\bvec y}\in\complex^n$, if $\inprod{\bvec x}{\bvec y}$ is real, then we can write
$\inprod{\bvec x}{\bvec y}$ as $\frac{1}{2}[\norm{{\bvec x}}^2+\norm{{\bvec y}}^2-\norm{\bvec{x-y}}^2]$. Then, defining $\alpha'=\alpha/B'$, we can write,
\begin{align*}
\lefteqn{F_1 - \frac{\norm{\dualproj{T^{k-1}}\bvec{\Phi}\curly{S^{k-1}}\bvec{x}\curly{S^{k-1}}}^2}{\sqrt{\alpha'}\norm{\bvec{x}\curly{S^{k-1}}}}} & &\\
\ & \stackrel{(e)}{=} \frac{1}{\sqrt{\alpha'}\norm{\bvec{x}\curly{S^{k-1}}}}\left(\frac{\alpha'\norm{\bvec{x}\curly{S^{k-1}}}^2\norm{\dualproj{T^{k-1}}\bvec{\phi}_{j_1}}^2}{4}\right.\\
\ &\left. - \norm{\frac{\sqrt{\alpha'}\norm{\bvec{x}\curly{S^{k-1}}}\dualproj{T^{k-1}}\bvec{\phi}_{j_1}s_{j_1}}{2} - \bvec{q}^{k-1}}^2\right)\nonumber\\
\ & = \frac{1}{\sqrt{\alpha'}\norm{\bvec{x}\curly{S^{k-1}}}}\left(\frac{\alpha'\norm{\bvec{x}\curly{S^{k-1}}}^2\norm{\dualproj{T^{k-1}}\bvec{\phi}_{j_1}}^2}{4}\right.\\
\ & -\left. \norm{\dualproj{T^{k-1}}\bvec{\Phi}\curly{S^{k-1}}\bvec{v}\curly{S^{k-1}}}^2\right),\nonumber
\end{align*}
where the vector $\bvec{v}\in \complex^n$ is defined as follows : $v_r = 0\ \forall r\notin I(\bvec{\Phi}\curly{S^{k-1}})$  and $\bvec{v}\curly{S^{k-1}}=
\bvec{x}\curly{S^{k-1}}$, except for the index $j_1$, for which, $v_{j_1} = x_{j_1} - \frac{s_{j_1}\sqrt{\alpha'}\norm{\bvec{x}\curly{S^{k-1}}}}{2}$. 
Then, as $j_1\notin I(\bvec{\Phi}\curly{T^{k-1}})$, using Lemma~\ref{lem:projected-column-lower-bound-wielandt-PIBS} we obtain \begin{align}
\norm{\dualproj{T^{k-1}}\bvec{\phi}_{j_1}}^2 & \ge (1-\delta^2_{b,p,1,L'}(\abs{T^{k-1}},1)).
\end{align} On the other hand, as the support of $\bvec{v}$ is $I(\bvec{\Phi}\curly{S^{k-1}})$ and since $I(\bvec{\Phi}\curly{T^{k-1}})\cup I(\bvec{\Phi}\curly{S^{k-1}})$ is the support of a PIBS vector with parameters $(b,p,1,L',K,0)$, using Lemma~\ref{lem:projected-matrix-pibrip}, we obtain \begin{align}
\norm{\dualproj{T^{k-1}}\bvec{\Phi}\curly{S^{k-1}}\bvec{v}\curly{S^{k-1}}}^2 & \le (1+\delta_{b,p,1,L'}(K,0))\norm{\bvec{v}}^2.
\end{align} Therefore, we obtain, \begin{align}
\lefteqn{F_1 - \frac{\norm{\dualproj{T^{k-1}}\bvec{\Phi}\curly{S^{k-1}}\bvec{x}\curly{S^{k-1}}}^2}{\sqrt{\alpha'}\norm{\bvec{x}\curly{S^{k-1}}}}} & &\nonumber\\
\ge & \frac{1}{\sqrt{\alpha'}\norm{\bvec{x}\curly{S^{k-1}}}}\left(\alpha'/4\norm{\bvec{x}\curly{S^{k-1}}}^2(1-\delta^2_{b,p,1,L'}(\abs{T^{k-1}},1))\right.\nonumber\\
\label{eq:step-k-block-selection-lower-bound-intermediate}
\ & - \left.(1+\delta_{b,p,1,L'}(K,0))\norm{\bvec{v}}^2\right).
\end{align}
Now observe that \begin{align}
\label{eq:v-expression}
\norm{\bvec{v}}^2 & = \norm{\bvec{x}\curly{S^{k-1}}}^2 - \abs{x_{j_1}}^2 + \abs{x_{j_1} - \frac{s_{j_1}\sqrt{\alpha'}\norm{\bvec{x}\curly{S^{k-1}}}}{2}}^2 \nonumber\\
\ & = \left(1+\frac{\alpha'}{4}\right) \norm{\bvec{x}\curly{S^{k-1}}}^2 - \sqrt{\alpha'}\Re (s_{j_1}^*x_{j_1}) \norm{\bvec{x}\curly{S^{k-1}}}.
\end{align}

Therefore, from the inequalities~\eqref{eq:step-k-block-selection-alternative-criterion-lhs-preliminary-lower-bound},~\eqref{eq:F2_upper_bound},~\eqref{eq:F3-upper-bound},~\eqref{eq:step-k-block-selection-lower-bound-intermediate}, and~\eqref{eq:v-expression}, one can deduce that, \begin{align}
\lefteqn{\abs{\inprod{\bvec{\phi}_{j_1}}{\bvec{r}^{k-1}}} - \abs{\inprod{\bvec{\phi}_{j_2}}{\bvec{r}^{k-1}}}}  & &\nonumber\\
\ge & \frac{\norm{\bvec{x}\curly{S^{k-1}}}}{\sqrt{\alpha'}}(1-\sqrt{\alpha + 1}\delta_{b,p,1,L'}(K,1))\nonumber\\
\ & + \frac{1}{\sqrt{\alpha'}\norm{\bvec{x}\curly{S^{k-1}}}}\left[\frac{\alpha'}{4}\norm{\bvec{x}\curly{S^{k-1}}}^2\right.\nonumber\\
\ & (1-\delta^2_{b,p,1,L'}(\abs{T^{k-1}},1))\nonumber\\
\ & - \left.(1+\delta_{b,p,1,L'}(K,0))\left\{\left(1+\frac{\alpha'}{4}\right) \norm{\bvec{x}\curly{S^{k-1}}}^2 \right.\right.\nonumber\\
\ & - \left.\left.\sqrt{\alpha'}\Re(s_{j_1}^*x_{j_1}) \norm{\bvec{x}\curly{S^{k-1}}}\right\}\right]\nonumber\\
\ & - \sqrt{(1+B')(1+\delta_{b,p,1,L'}(\abs{T^{k-1}},2))}\epsilon\nonumber\\
\ & \ge - \frac{\norm{\bvec{x}\curly{S^{k-1}}}}{\sqrt{\alpha'}}\left[\frac{\alpha'}{4}\left(\delta^2_{b,p,1,L'}(\abs{T^{k-1}},1)\right.\right.\nonumber\\
\ & \left.\left.+\delta_{b,p,1,L'}(K,0)\right) + \sqrt{\alpha + 1}\delta_{b,p,1,L'}(K,1)\right.\nonumber\\
\ & \left.+\delta_{b,p,1,L'}(K,0)\right] + \min_{j_1\in C^k}\Re (s_{j_1}^*x_{j_1})(1+\delta_{b,p,1,L'}(K,0))\nonumber\\
\label{eq:step-k-block-selection-preliminary-sufficient-condition}
\ & - \sqrt{(1+B')(1+\delta_{b,p,1,L'}(\abs{T^{k-1}},2))}\epsilon.
\end{align}
Clearly, the condition~\eqref{eq:step-k-block-selection-alternative-criterion} is satisfied if the right hand side of the inequality~\eqref{eq:step-k-block-selection-preliminary-sufficient-condition} is non-negative. However, observe that, this is not possible unless $\min_{j_1\in C^k}\Re (s_{j_1}^*x_{j_1})>0$. Now we have, \begin{align}
s_{j_1} & = \sign{\inprod{\bvec{\phi}_{j_1}}{\dualproj{T^{k-1}}\bvec{\Phi}\curly{S^{k-1}}\bvec{x}\curly{S^{k-1}}}}\nonumber\\
\ & = \sign{\norm{\dualproj{T^{k-1}}\bvec{\phi}_{j_1}}^2x_{j_1} + \inprod{\bvec{\phi}_{j_1}}{\dualproj{T^{k-1}}\bvec{\Phi}_{V^{k-1}_{j_1}}\bvec{x}_{V^{k-1}_{j_1}}}}\nonumber\\
\ & =\sign{z_{j_1}+w_{j_1}},
\end{align} where $ z_{j_1}=\norm{\dualproj{T^{k-1}}\bvec{\phi}_{j_1}}^2x_{j_1},\ w_{j_1}=\inprod{\bvec{\phi}_{j_1}}{\dualproj{T^{k-1}}\bvec{\Phi}_{V^{k-1}_{j_1}}\bvec{x}_{V^{k-1}_{j_1}}}$. Clearly, $\Re(s_{j_1}^\star x_{j_1})>0,\forall j_1\in C^k$, if and only if $\forall j_1\in C^k$ \begin{align}
\Re((z_{j_1}+w_{j_1})^*z_{j_1}) = \abs{z_{j_1}}^2 + \Re(w_{j_1}^\star z_{j_1})>0.
\end{align} Now, as $\Re(w_{j_1}^\star z_{j_1})\ge -\abs{w_{j_1}}\abs{z_{j_1}}$, the above is satisfied if $ \abs{z_{j_1}}^2-\abs{w_{j_1}}\abs{z_{j_1}}>0$. Therefore, a sufficient condition for the above to be satisfied is the following: \begin{align}
\abs{z_{j_1}} & >\abs{w_{j_1}},\ \forall j_1\in C^k,\nonumber\\
\Leftrightarrow   \norm{\dualproj{T^{k-1}}\bvec{\phi}_{j_1}}^2\abs{x_{j_1}} & > \abs{\inprod{\bvec{\phi}_{j_1}}{\dualproj{T^{k-1}}\bvec{\Phi}_{V^{k-1}_{j_1}}\bvec{x}_{V^{k-1}_{j_1}}}},
\end{align} $\forall j_1\in C^k.$
%
We have already found during the derivation of inequality~\eqref{eq:step-k-block-selection-lower-bound-intermediate} that $\norm{\dualproj{T^{k-1}}\bvec{\phi}_{j_1}}^2\ge (1-\delta^2_{b,p,1,L'}(\abs{T^{k-1}},1))$. Moreover, to find an upper bound on $\abs{\inprod{\bvec{\phi}_{j_1}}{\dualproj{T^{k-1}}\bvec{\Phi}_{V^{k-1}_{j_1}}\bvec{x}_{V^{k-1}_{j_1}}}}$, first note that the set $V_{j_1}^{k-1}$ is disjoint to the index $j_1$, and note that the union of the sets $I(\bvec{\Phi}\curly{T^{k-1}})$ and $V_{j_1}^{k-1}$ and the index $j_1$ corresponds to the union of the sets $I(\bvec{\Phi}\curly{T^{k-1}})$ and $I(\bvec{\Phi}\curly{S^{k-1}})$, which corresponds to the support of a PIBS vector with parameters $(b,p,1,L',K,0)$. Therefore, using Lemma~\ref{lem:projected-matrix-innerproduct-inequality} we obtain $\abs{\inprod{\bvec{\phi}_{j_1}}{\dualproj{T^{k-1}}\bvec{\Phi}_{V^{k-1}_{j_1}}\bvec{x}_{V^{k-1}_{j_1}}}}\le \delta_{b,p,1,L'}(K,0)\norm{\bvec{x}_{V^{k-1}_{j_1}}}$. Hence, if for all $j_1\in C^k$, $(1-\delta^2_{b,p,1,L'}(\abs{T^{k-1}},1))\abs{x_{j_1}}>\delta_{b,p,1,L'}(K,0)\norm{\bvec{x}_{V^{k-1}_{j_1}}}$ holds, then $\Re (s_{j_1}^\star x_{j_1})>0,\ \forall j_1\in C^k$. Since $\norm{\bvec{x}_{V^{k-1}_{j_1}}}=\sqrt{\norm{\bvec{x}\curly{S^{k-1}}}^2 - \abs{x_{j_1}}^2}$, we further deduce that $\Re(s_{j_1}^*x_{j_1})>0$ if $(1-\delta^2_{b,p,1,L'}(\abs{T^{k-1}},1))\abs{x_{j_1}} >\delta_{b,p,1,L'}(K,0)\sqrt{\norm{\bvec{x}\curly{S^{k-1}}}^2 - \abs{x_{j_1}}^2}$, which is equivalent to \begin{align}
\frac{\abs{x_{j_1}}}{\frac{\norm{\bvec{x}\curly{S^{k-1}}}}{\sqrt{c_k b}}} >\frac{\delta_{b,p,1,L'}(K,0)\sqrt{c_k b}}{\sqrt{(1-\delta^2_{b,p,1,L'}(\abs{T^{k-1}},1))^2 + \delta^2_{b,p,1,L'}(K,0)}},
\end{align} for all $j_1\in C^k$. Now, let $x_{\min,k-1}=\min\{\abs{x_j}:\ j\in I(\bvec{\Phi}\curly{S^{k-1}})\}\le \abs{x_{j_1}}$, and $x_{\max, k-1}:=\max\{\abs{x_j}:\ j\in I(\bvec{\Phi}\curly{S^{k-1}})\}\ge \frac{\norm{\bvec{x}\curly{S^{k-1}}}}{\sqrt{c_k b}}$. Thus, to ensure that $\Re(s_{j_1}^*x_{j_1})>0$ for all $j_1\in C^k$, the following serves as a sufficient condition: \begin{align}
\label{eq:step-k-min-to-max-ratio-sign-condition}
\frac{x_{\min,k-1}}{x_{\max,k-1}}>\frac{\delta_{b,p,1,L'}(K,0)\sqrt{c_k b}}{\sqrt{(1-\delta^2_{b,p,1,L'}(\abs{T^{k-1}},1))^2 + \delta^2_{b,p,1,L'}(K,0)}},
\end{align}
where 

Now to find a lower bound on $\Re(s_{j_1}^*x_{j_1})$ under the condition that $\abs{z_{j_1}}>\abs{w_{j_1}},\ \forall j_1\in C^k$, we claim the following:
\begin{lem}
	\label{lem:lower-bound-real-part}
	Let $w,z\in \complex$ such that $z,w\ne 0$. Let $v=w/z$. Then, if $\abs{v}<1$, the following holds true:
	\begin{align}
	\Re\left(\frac{(z+w)}{\abs{z+w}}z^*\right) & \left\{
	\begin{array}{ll}
	=\abs{z}, & \mbox{if $v$ is real},\\
	\ge \abs{z}\sqrt{1-\abs{v}^2}.
	\end{array}
	\right.
	\end{align}
\end{lem}
\begin{proof}
	See Appendix~\ref{sec:appendix-proof-lemma-lower-bound-real-part} for a proof of this claim.
\end{proof}
Therefore, if $\bvec{x},\bvec{\Phi}$ have real valued entries, under the condition $\norm{\dualproj{T^{k-1}}\bvec{\phi}_{j_1}}^2\abs{x_{j_1}}>\abs{\inprod{\bvec{\phi}_{j_1}}{\dualproj{T^{k-1}}\bvec{\Phi}_{V_{j_1}^{k-1}}\bvec{x}_{V_{j_1}^{k-1}}}}$, one has $\Re(s_{j_1}^*x_{j_1})=\abs{x_{j_1}}$. Consequently, when $\bvec{x},\bvec{\Phi}$ have real entries, condition~\eqref{eq:step-k-block-selection-alternative-criterion} is satisfied if the condition~\eqref{eq:step-k-min-to-max-ratio-sign-condition} is satisfied and if the right hand side of inequality~\eqref{eq:step-k-block-selection-preliminary-sufficient-condition} is made greater than or equal to $0$, which is ensured by the following sufficient condition:\begin{align}
\lefteqn{x_{\min,k-1}} & &\nonumber\\
\ & \ge \frac{\norm{\bvec{x}\curly{S^{k-1}}}}{\sqrt{\alpha'}(1+\delta_{b,p,1,L'}(K,0))}\left[\frac{\alpha'}{4}\left(\delta^2_{b,p,1,L'}(\abs{T^{k-1}},1)\right.\right.\nonumber\\
\ & \left. \left. +\delta_{b,p,1,L'}(K,0)\right)+ \sqrt{\alpha + 1}\delta_{b,p,1,L'}(K,1)+\delta_{b,p,1,L'}(K,0)\right]\nonumber\\
\ & + \frac{\sqrt{(1+B')(1+\delta_{b,p,1,L'}(\abs{T^{k-1}},2))}\epsilon}{(1+\delta_{b,p,1,L'}(K,0))}.
\end{align}
The above in turn is ensured by the following sufficient condition:
\begin{align}
\lefteqn{x_{\min,k-1}} & &\nonumber\\
\ & \ge \frac{x_{\max,k-1}\sqrt{c_kB'b}}{\sqrt{\alpha}(1+\delta_{b,p,1,L'}(K,0))}\left[\frac{\alpha'}{4}\left(\delta^2_{b,p,1,L'}(\abs{T^{k-1}},1)\right.\right.\nonumber\\
\ & \left. \left. +\delta_{b,p,1,L'}(K,0)\right)+ \sqrt{\alpha + 1}\delta_{b,p,1,L'}(K,1)+\delta_{b,p,1,L'}(K,0)\right]\nonumber\\
\label{eq:step-k-block-recovery-sufficient-condition-prelim}
\ & + \frac{\sqrt{(1+B')(1+\delta_{b,p,1,L'}(\abs{T^{k-1}},2))}\epsilon}{(1+\delta_{b,p,1,L'}(K,0))}.
\end{align}
As $\alpha$ can be any arbitrary positive number, choosing $\alpha=c_k$ and defining $c_k'=c_k/B$, we arrive at the following sufficient condition: 
\begin{align}
\lefteqn{x_{\min,k-1}} & &\nonumber\\
\ & \ge \frac{x_{\max,k-1}\sqrt{B'b}}{(1+\delta_{b,p,1,L'}(K,0))}\left[\frac{c_k'}{4}\left(\delta^2_{b,p,1,L'}(\abs{T^{k-1}},1)\right.\right.\nonumber\\
\ & \left. \left. +\delta_{b,p,1,L'}(K,0)\right)+ \sqrt{c_k + 1}\delta_{b,p,1,L'}(K,1)+\delta_{b,p,1,L'}(K,0)\right]\nonumber\\
\label{eq:step-k-block-recovery-sufficient-condition}
\ & + \frac{\sqrt{(1+B')(1+\delta_{b,p,1,L'}(\abs{T^{k-1}},2))}\epsilon}{(1+\delta_{b,p,1,L'}(K,0))}.
\end{align}
On the other hand, for general complex entries, we have, from Lemma~\ref{lem:lower-bound-real-part}, \begin{align}
\lefteqn{\Re(s_{j_1}^*x_{j_1}) \ge \abs{x_{j_1}}\sqrt{1-\left(\frac{\abs{\inprod{\bvec{\phi}_{j_1}}{\dualproj{T^{k-1}}\bvec{\Phi}_{V_{j_1}^{k-1}}\bvec{x}_{V_{j_1}^{k-1}}}}}{\norm{\dualproj{T^{k-1}}\bvec{\phi}_{j_1}}^2\abs{x_{j_1}}}\right)^2}} & &\nonumber\\
\ & \ge \sqrt{\abs{x_{j_1}}^2-\left(\frac{\delta_{b,p,1,L'}(K,0)}{1-\delta^2_{b,p,1,L'}(\abs{T^{k-1},1})}\right)^2\left(\norm{\bvec{x}\curly{S^{k-1}}}^2-\abs{x_{j_1}}^2\right)},
\end{align}
where the last step uses the lower bound of $\norm{\dualproj{T^{k-1}}\bvec{\phi}_{j_1}}^2$ and the upper bound of $\abs{\inprod{\bvec{\phi}_{j_1}}{\dualproj{T^{k-1}}\bvec{\Phi}_{V_{j_1}^{k-1}}\bvec{x}_{V_{j_1}^{k-1}}}}$.
Therefore, for general complex entries, a sufficient condition to make the RHS of inequality \eqref{eq:step-k-block-selection-preliminary-sufficient-condition} non-negative, is obtained (choosing $\alpha=c_k$ in the RHS of the inequality~\eqref{eq:step-k-block-selection-preliminary-sufficient-condition}) by ensuring the following for all $j_1\in C^k$: \begin{align}
\lefteqn{\sqrt{\abs{x_{j_1}}^2-\eta_{k-1}^2\left(\norm{\bvec{x}\curly{S^{k-1}}}^2-\abs{x_{j_1}}^2\right)}} & & \nonumber\\
\ & \ge \frac{\norm{\bvec{x}\curly{S^{k-1}}}}{\sqrt{c_k'}(1+\delta_{b,p,1,L'}(K,0))}\left[\frac{c_k'}{4}\left(\delta^2_{b,p,1,L'}(\abs{T^{k-1}},1)\right.\right.\nonumber\\
\ & \left.\left. +\delta_{b,p,1,L'}(K,0)\right)+ \sqrt{c_k + 1}\delta_{b,p,1,L'}(K,1)+\delta_{b,p,1,L'}(K,0)\right]\nonumber\\
\ & + \frac{\sqrt{(1+B')(1+\delta_{b,p,1,L'}(\abs{T^{k-1}},2))}\epsilon}{1+\delta_{b,p,1,L'}(K,0)},
\end{align}
where $\eta_{k-1} = \left(\frac{\delta_{b,p,1,L'}(K,0)}{1-\delta^2_{b,p,1,L'}(\abs{T^{k-1},1})}\right)$. The above, in turn is ensured if \begin{align}
\lefteqn{x_{\min,k-1}^2} & &\nonumber\\
\ & -\left(\frac{\delta_{b,p,1,L'}(K,0)}{1-\delta^2_{b,p,1,L'}(\abs{T^{k-1},1})}\right)^2\left(c_kb x_{\max,k-1}^2-x_{\min,k-1}^2\right) & & \nonumber\\
\ & \ge \left(\frac{x_{\max,k-1}\sqrt{B'b}}{1+\delta_{b,p,1,L'}(K,0)}\left[\frac{c_k'}{4}\left(\delta^2_{b,p,1,L'}(\abs{T^{k-1}},1)\right.\right.\right.\nonumber\\
\ & \left. \left. +\delta_{b,p,1,L'}(K,0)\right)+ \sqrt{c_k + 1}\delta_{b,p,1,L'}(K,1)+\delta_{b,p,1,L'}(K,0)\right]\nonumber\\
\label{eq:step-k-block-recovery-sufficient-condition-complex}
\ & \left.+ \frac{\sqrt{(1+B')(1+\delta_{b,p,1,L'}(\abs{T^{k-1}},2))}\epsilon}{1+\delta_{b,p,1,L'}(K,0)}\right)^2.
\end{align}

\subsection{Condition for overall success}
\label{sec:condition-overall-success}
In this section we claim that the conditions~\eqref{eq:tsgbomp-delta-recovery-condition}, and~\eqref{eq:tsgbomp-noisy-recovery-condition} stated in Theorem~\ref{thm:recovery-guarantee-tsgbomp} simultaneously satisfy the inequalities~\eqref{eq:step-k-sufficient-condition-correct-window-selection},~\eqref{eq:step-k-min-to-max-ratio-sign-condition}, and~\eqref{eq:step-k-block-recovery-sufficient-condition} for all iterations $1\le k\le K$. In the following $\delta$ will be used to denote $\delta_{b,p,L',L'}(K-1,2)$ unless otherwise specified. 

We first verify that the inequality~\eqref{eq:step-k-sufficient-condition-correct-window-selection} is satisfied for all $1\le k\le K$ under the conditions~\eqref{eq:tsgbomp-delta-recovery-condition}, and~\eqref{eq:tsgbomp-noisy-recovery-condition}. Indeed, using Lemmas~\ref{lem:PIBRIP-B-L-monotonicity} and~\ref{lem:PIBRIP-K-R-monotonicity} we obtain $\delta_{b,p,L',L'}(K,1)\le \delta_{b,p,L',L'}(K-1,2)$, and $\delta_{b,p,L',L'}(\abs{T^{k-1}},2)\le \delta_{b,p,L',L'}(K-1,2)$, respectively. Moreover, for all $1\le k\le K$, $d_k=\abs{O_{S^{k-1}}}\le 2\abs{S^{k-1}}=2c_k\le 2\abs{T}=2K$, since each cluster in $S^{k-1}$ might have non-zero overlap with at most two windows in $O_{S^{k-1}}$. Therefore, $\frac{\norm{\bvec{x}\curly{S^{k-1}}}}{\sqrt{d_k}}\ge \frac{\norm{\bvec{x}\curly{S^{k-1}}}}{\sqrt{2c_k}}\ge \frac{\sqrt{b}x_{\min}}{\sqrt{2}}\ge \frac{x_{\min}}{\sqrt{2}}>\frac{\sqrt{(1+B')(1+\delta_{b,p,L',L'}(\abs{T^{k-1}},2))}\epsilon}{(1-\delta_{b,p,L',L'}(K,1)\sqrt{d_k+1})}\ge \frac{\sqrt{2(1+\delta_{b,p,L',L'}(\abs{T^{k-1}},2))}\epsilon}{(1-\delta_{b,p,L',L'}(K,1)\sqrt{d_k+1})}$, where the last two inequalities follow from condition~\eqref{eq:tsgbomp-noisy-recovery-condition} and the fact that $B'\ge 1$. After rearrangement, this results in the condition~\eqref{eq:step-k-sufficient-condition-correct-window-selection}.

Now, to show that~\eqref{eq:tsgbomp-delta-recovery-condition} and \eqref{eq:tsgbomp-noisy-recovery-condition} imply~\eqref{eq:step-k-min-to-max-ratio-sign-condition}, we first observe that \begin{align*}
\lefteqn{\frac{K}{4\sqrt{B'}}+\frac{\sqrt{B'}(\sqrt{K+1}+1)}{1+\delta} - \frac{\sqrt{K}}{\sqrt{\delta^2+(1-\delta^2)^2}}} & &\\
\ & \stackrel{(g)}{\ge} \frac{K}{4\sqrt{B'}}+\frac{\sqrt{2B'}(\sqrt{K+1}+1)}{\sqrt{2}+1}-\sqrt{\frac{4K}{3}}\\
\ & = \frac{1}{\sqrt{B'}}\left(\frac{\sqrt{K}}{2}-\frac{2\sqrt{B'}}{\sqrt{3}}\right)^2\\
\ & + \frac{\sqrt{2B'}}{\sqrt{2}+1}\left(\sqrt{K+1} - \frac{2\sqrt{2}+1}{3}\right)\stackrel{(h)}{>}0,
\end{align*}
where step $(g)$ uses $\delta<\frac{1}{\sqrt{2}}$ (which follows from~\eqref{eq:tsgbomp-delta-recovery-condition} since $K\ge 1$), and the fact that the function $\delta^2+(1-\delta^2)^2$ is monotonically decreasing for $\delta\in [0,\frac{1}{\sqrt{2}}]$ with the minimum at $\frac{1}{\sqrt{2}}$, and step $(h)$ uses the simple observation that $\frac{2\sqrt{2}+1}{3}<\sqrt{2}\le \sqrt{K+1}$ since $K\ge 1$. Then, using~\eqref{eq:tsgbomp-noisy-recovery-condition} and~\eqref{eq:tsgbomp-delta-recovery-condition} it follows that \begin{align}
\ & x_{\min,k-1}\ge x_{\min}> \frac{x_{\max}\delta\sqrt{Kb}}{\sqrt{\delta^2+(1-\delta^2)^2}}\nonumber\\
\ & \ge \frac{x_{\max,k-1}\delta_{b,p,1,L'}(K,0)\sqrt{c_kb}}{\sqrt{(1-\delta^2_{b,p,1,L'}(\abs{T^{k-1}},1))^2+\delta^2_{b,p,1,L'}(K,0)}}.
\end{align} In the above we use the fact that the function $\frac{\delta}{\sqrt{(1-\delta^2)^2+\delta^2}}$ is monotonically increasing for $\delta\in [0,1]$, so that, using $x_{\max}\ge x_{\max,k-1},\ c_k\le K$, and the Lemmas~\ref{lem:PIBRIP-K-R-monotonicity},~\ref{lem:PIBRIP-L-monotonicity} and~\ref{lem:PIBRIP-B-L-monotonicity}, we obtain, $\frac{\delta}{\sqrt{(1-\delta^2)^2+\delta^2}}\ge \frac{\delta_{b,p,1,L'}(K,1)}{\sqrt{(1-\delta^2_{b,p,1,L'}(K,1))^2+\delta^2_{b,p,1,L'}(K,1)}}=\frac{1}{\sqrt{1+\left(\frac{1-\delta^2_{b,p,1,L'}(K,1)}{\delta_{b,p,1,L'}(K,1)}\right)^2}}$. Finally, using $\delta_{b,p,1,L'}(K,1)\ge \delta_{b,p,1,L'}(\abs{T^{k-1}},1)$, and $\delta_{b,p,1,L'}(K,1)\ge \delta_{b,p,1,L'}(K,0)$, we obtain the above inequality.
 
To show that~\eqref{eq:step-k-block-recovery-sufficient-condition} follows from~\eqref{eq:tsgbomp-noisy-recovery-condition} and \eqref{eq:tsgbomp-delta-recovery-condition}, first observe that the condition~\eqref{eq:tsgbomp-noisy-recovery-condition} followed by~\eqref{eq:tsgbomp-delta-recovery-condition} imply that \begin{align}
\lefteqn{x_{\min}} & &\nonumber\\
\ & >x_{\max}\delta\sqrt{B'b}\left[\frac{K}{4B'}(1+\delta)+\sqrt{K+1}+1\right]-x_{\min}\delta\nonumber\\
\ & +\epsilon\sqrt{(1+B')(1+\delta)}\nonumber\\
\ & \ge \left[x_{\max,k-1}\sqrt{B'b}\left(\frac{c_k}{4B'}+1\right)-x_{\min,k-1}\right]\delta\nonumber\\
\ & +x_{\max,k-1}\sqrt{B'b}\left[\frac{c_k}{4B'}\delta^2+\delta\sqrt{c_k+1}\right]\nonumber\\
\label{eq:preliminary-verfiication-inequality1-real-case}
\ & +\epsilon\sqrt{(1+B')(1+\delta)},
\end{align} since $x_{\max}\ge x_{\max,k-1}\ge x_{\min,k-1}\ge x_{\min}$ and $K\ge c_k$. Note that factor multiplied with $\delta$ in the first term in the RHS above is non-negative. Therefore, applying Lemmas~\ref{lem:PIBRIP-K-R-monotonicity},~\ref{lem:PIBRIP-B-L-monotonicity} and~\ref{lem:PIBRIP-L-monotonicity} it follows that, \begin{align}
\lefteqn{x_{\min,k-1}} & & \nonumber\\
\ & > \left[x_{\max,k-1}\sqrt{B'b}\left(\frac{c_k}{4B'}+1\right)-x_{\min,k-1}\right]\delta_{b,p,1,L'}(K,0)\nonumber\\
\ & +x_{\max,k-1}\sqrt{B'b}\left[\frac{c_k}{4B'}\delta_{b,p,1,L'}^2(\abs{T^{k-1}},1)\right.\nonumber\\
\ & \left.+\delta_{b,p,1,L'}(K,1)\sqrt{c_k+1}\right]\nonumber\\
\label{eq:preliminary-verfiication-inequality2-real-case}
\ & +\epsilon\sqrt{(1+B')(1+\delta_{b,p,1,L'}(\abs{T^{k-1}},2))},
\end{align}
which, after a rearrangement, is identical to the condition~\eqref{eq:step-k-block-recovery-sufficient-condition}. In the above we have used the following inequalities to arrive at the final inequality: $\delta=\delta_{b,p,L',L'}(K-1,2)\ge \delta_{b,p,L',L'}(K,1)$(Lemma~\ref{lem:PIBRIP-B-L-monotonicity}), $\delta_{b,p,L',L'}(K,1)\ge \delta_{b,p,1,L'}(K,0)$(Lemmas~\ref{lem:PIBRIP-L-monotonicity} and~\ref{lem:PIBRIP-K-R-monotonicity}), $\delta_{b,p,L',L'}(K,1)\ge \delta_{b,p,L',L'}(\abs{T^{k-1}},1)\ge \delta_{b,p,1,L'}(\abs{T^{k-1}},1)$ (by Lemma~\ref{lem:PIBRIP-K-R-monotonicity} followed by Lemma~\ref{lem:PIBRIP-L-monotonicity}), and finally, $\delta_{b,p,L',L'}(K-1,2)\ge \delta_{b,p,1,L'}(K-1,2)\ge\delta_{b,p,1,L'}(\abs{T^{k-1}},2) $(by Lemma~\ref{lem:PIBRIP-L-monotonicity} followed by Lemma~\ref{lem:PIBRIP-K-R-monotonicity}).

We now proceed to prove that the condition~\eqref{eq:tsgbomp-delta-recovery-condition-complex} along with~\eqref{eq:tsgbomp-noisy-recovery-condition-complex} are sufficient to ensure that the conditions~\eqref{eq:step-k-sufficient-condition-correct-window-selection} and~\eqref{eq:step-k-block-recovery-sufficient-condition-complex} are simultaneously satisfied. First observe that one can readily verify that the condition~\eqref{eq:step-k-sufficient-condition-correct-window-selection} is satisfied using the same arguments as in the real case, as in this case too we have the inequality $x_{\min}>\frac{\epsilon\sqrt{2(1+B')(1+\delta_{b,p,L',L'}(\abs{T^{k-1}},2))}}{1-\delta_{b,p,L',L'}(K,1)\sqrt{d_k+1}}.$ Therefore, it just remains to verify that the inequality~\eqref{eq:step-k-block-recovery-sufficient-condition-complex} follows.
Observe that, with $\delta=\delta_{b,p,L',L'}(K-1,2)$, inequality~\eqref{eq:tsgbomp-noisy-recovery-condition-complex} implies \begin{align}
\lefteqn{x_{\min}\sqrt{1+\frac{\delta^2}{(1-\delta^2)^2}}} & &\nonumber\\
\ & > \left[\frac{\delta \sqrt{Kb}}{1-\delta^2}+\frac{\delta\sqrt{B'b}}{1+\delta}\left(\frac{K(1+\delta)}{4B'}+\sqrt{K+1}+1\right)\right]x_{\max}\nonumber\\
\ & + \frac{\epsilon\sqrt{(1+B')(1+\delta)}}{1+\delta}.
\end{align}
Now, using the inequality $ a+b\ge \sqrt{a^2+b^2},\ a,b\ge 0$ in the RHS of the above inequality, we see that the above inequality implies the following: \begin{align}
\lefteqn{x_{\min}^2\left(1+\frac{\delta^2}{(1-\delta^2)^2}\right) > \frac{\delta^2 Kb}{(1-\delta^2)^2}x_{\max}^2} & &\nonumber\\
\ & + \left(\frac{\delta\sqrt{B'b}}{1+\delta}\left(\frac{K(1+\delta)}{4B'}+\sqrt{K+1}+1\right)x_{\max}+\right.\nonumber\\
\ & \left.\frac{\epsilon\sqrt{(1+B')(1+\delta)}}{1+\delta}\right)^2,
\end{align}
which implies \begin{align}
\lefteqn{\sqrt{x_{\min}^2 - \frac{\delta^2}{(1-\delta^2)^2}\left(Kbx_{\max}^2 - x_{\min}^2\right)}} & &\nonumber\\
\ & > \frac{\delta\sqrt{B'b}x_{\max}}{1+\delta}\left(\frac{K(1+\delta)}{4B'}+\sqrt{K+1}+1\right)\nonumber\\
\label{eq:preliminary-verification-inequality1-complex-case}
\ & +\frac{\epsilon\sqrt{(1+B')(1+\delta)}}{1+\delta}.
\end{align}
Let, for any $k\ge 1$, $z_{k-1}:=\sqrt{x_{\min,k-1}^2 - \frac{\delta^2}{(1-\delta^2)^2}\left(c_kbx_{\max,k-1}^2 - x_{\min,k-1}^2\right)}$, and that $z = \sqrt{x_{\min}^2 - \frac{\delta^2}{(1-\delta^2)^2}\left(Kbx_{\max}^2 - x_{\min}^2\right)}$, so that $z_{k-1}\ge z$, and $x_{\max,k-1}\ge z_{k-1}$ for all $k\ge 1$. Then, using exactly the same arguments used to arrive at inequality~\eqref{eq:preliminary-verfiication-inequality2-real-case} from inequality~\eqref{eq:preliminary-verfiication-inequality1-real-case}, it follows that the inequality~\eqref{eq:preliminary-verification-inequality1-complex-case} implies the following: \begin{align}
\lefteqn{z_{k-1} } & &\nonumber\\
\ & >\left[x_{\max,k-1}\sqrt{B'b}\left(\frac{c_k}{4B'}+1\right)-z_{k-1}\right]\delta_{b,p,1,L'}(K,0)\nonumber\\
\ & +x_{\max,k-1}\sqrt{B'b}\left[\frac{c_k}{4B'}\delta_{b,p,1,L'}^2(\abs{T^{k-1}},1)\right.\nonumber\\
\ & \left.+\delta_{b,p,1,L'}(K,1)\sqrt{c_k+1}\right]\nonumber\\
\label{eq:preliminary-verfiication-inequality2-complex-case}
\ & +\epsilon\sqrt{(1+B')(1+\delta_{b,p,1,L'}(\abs{T^{k-1}},2))},
\end{align}
which, after using the fact that $z_{k-1}\le \sqrt{x_{\min,k-1}^2 - \frac{\delta_{b,p,1,L'}^2(K,0)}{(1-\delta^2_{b,p,1,L'}(\abs{T^{k-1}},1)^2)^2}\left(c_kbx_{\max,k-1}^2 - x_{\min,k-1}^2\right)}$, followed by a rearrangement, yields the inequality~\eqref{eq:step-k-block-recovery-sufficient-condition-complex}.
\section{Random Matrices}
\label{sec:tsgbomp-random-matrices}
We now verify that one can choose the sizes $m,n$ of a random matrix along with a randomly generated signal with $(b,p,L,L',K,0)$ structure can indeed satisfy the condition~\eqref{eq:tsgbomp-noisy-recovery-condition}. We consider the noiseless case for deriving the conditions.
\begin{thm}
	\label{thm:random-matrix-guarantees}
	Let $x_{\min}=\min\{\abs{x_j}:j\in \supp(\bvec{x})\},\ x_{\max}=\max\{\abs{x_j}:j\in \supp(\bvec{x})\}$, $f_K:\real_+\to \real_+$, $f_K(u) = \frac{u\sqrt{B'b}}{1+u}\left(\frac{K(1+u)}{4B'}+\sqrt{K+1}+1\right)$, and $\bvec{\Phi}=[\phi_{ij}]_{m\times n}$ where $\phi_{ij}$'s are i.i.d.~$\mathcal{N}(0,m^{-1})$. 
	
	Let, $h(b,p,L',K,R)=A+KC+K\ln(pD/K-E),$ where
	\begin{align}
	\label{eq:A-expression}
	A & =\frac{3p(K-1)}{2(p+1)^2}+2,\\
	\label{eq:C-expression}
	C & =\ln p + 21/ 8 -1/p,\\
	\label{eq:D-expression}
	D & =n-(K-1)b+L',\\
	\label{eq:E-expression} 
	E & =L'-1.
	\end{align} 
	Also, let $\lambda = \sqrt{((K-1)b+2L)/m}$, $\nu=\lambda^2+2\lambda,\ \rho=f_K^{-1}(1)$, and $c_1=2e^{h(b,p,L,K,R)}, c_2=\frac{m}{(\lambda+1+\sqrt{1+\rho})^2}$ and let there is $\epsilon_0>0$ such that $\nu+\epsilon_0<\frac{1}{\sqrt{2K+1}}<\rho$. 
	Then, if the following conditions hold: $n\ge Kb+RL'+(K+1)(L-1');\ 3p\le K; m\ge (K-1)b+2L';\nu<\rho<3$, then the conditions~\eqref{eq:tsgbomp-delta-recovery-condition} and~\eqref{eq:tsgbomp-noisy-recovery-condition} are simultaneously satisfied with probability exceeding \begin{align}
	\label{eq:probability-lower-bound}
	g(f_K(\nu+\epsilon)) - (1+g(f_K(\nu+\epsilon)))c_1e^{-c_2\epsilon_0^2},
	\end{align} 
	where $g(a): = \prob{\frac{x_{\min}}{x_{\max}}>a}\forall\ a\in[0,1]$.
	
	 Furthermore, if $\bvec{x}\sim\mathcal{N}(\bvec{0},\bvec{I})$,\begin{align}
	\label{eq:g-function-upper-lower-bounds}
	w^{Kb-1}(a)\le \frac{g(a)}{Kb}\le  w(a),
	\end{align}
	where $w(a)=\frac{2}{\pi}\cot^{-1}a - \frac{1}{2}$.
\end{thm}
\begin{proof}
	See Appendix~\ref{sec:appendix-proof-thm-random-martix-guarantees}.
\end{proof}
\section{Simulation results}
\label{sec:simulation-results}
\begin{figure}[t!]
	\centering
	\begin{subfigure}{0.5\textwidth}
		\centering
		\includegraphics[height=1.5in,width=3in]{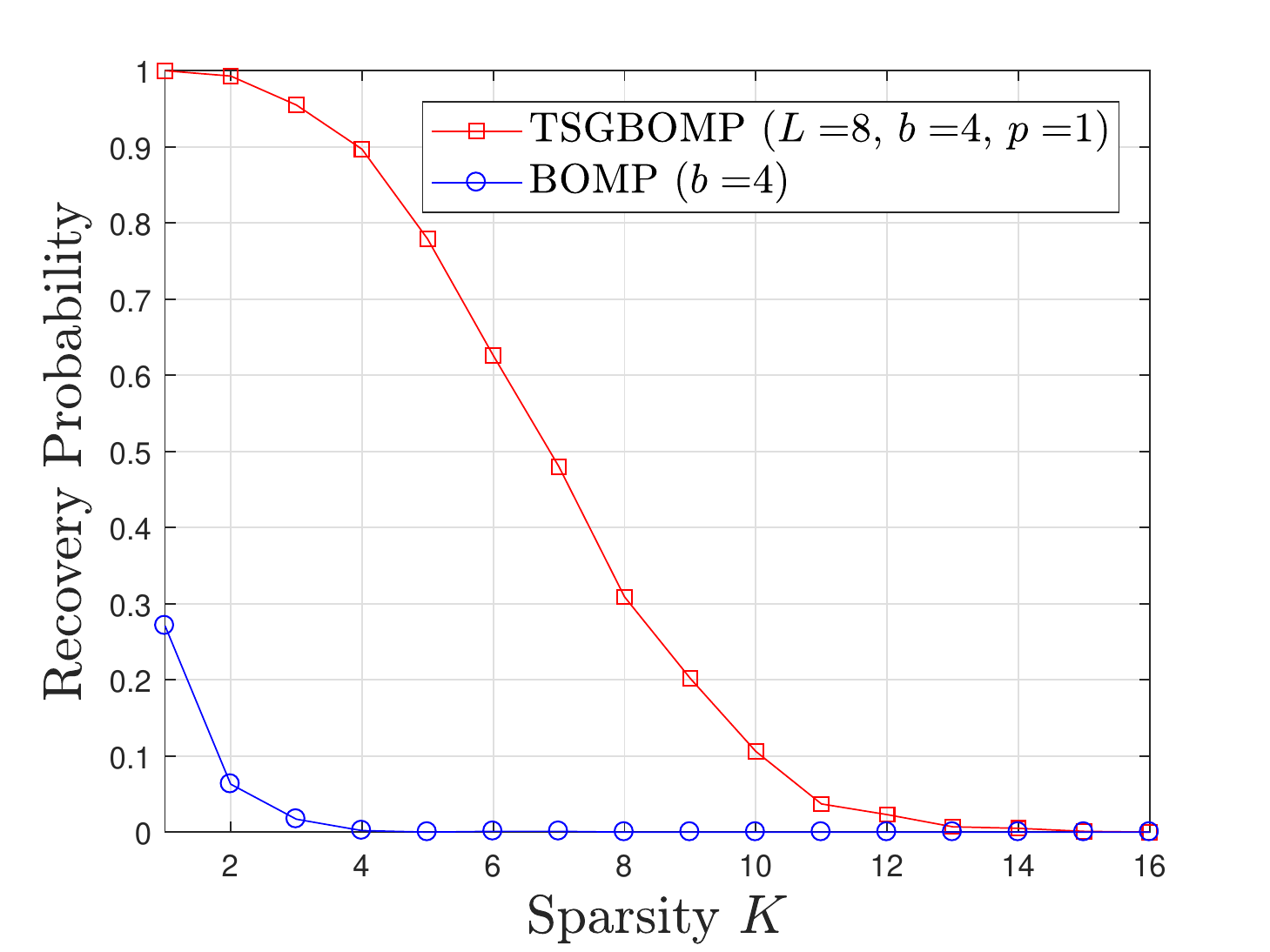}
		\caption{$p=1$}
		\label{fig:prob-recovery-vs-sparsity-m=120-n=200-p=1}
	\end{subfigure}
	\begin{subfigure}{0.5\textwidth}
		\centering
		\includegraphics[height=1.5in,width=3in]{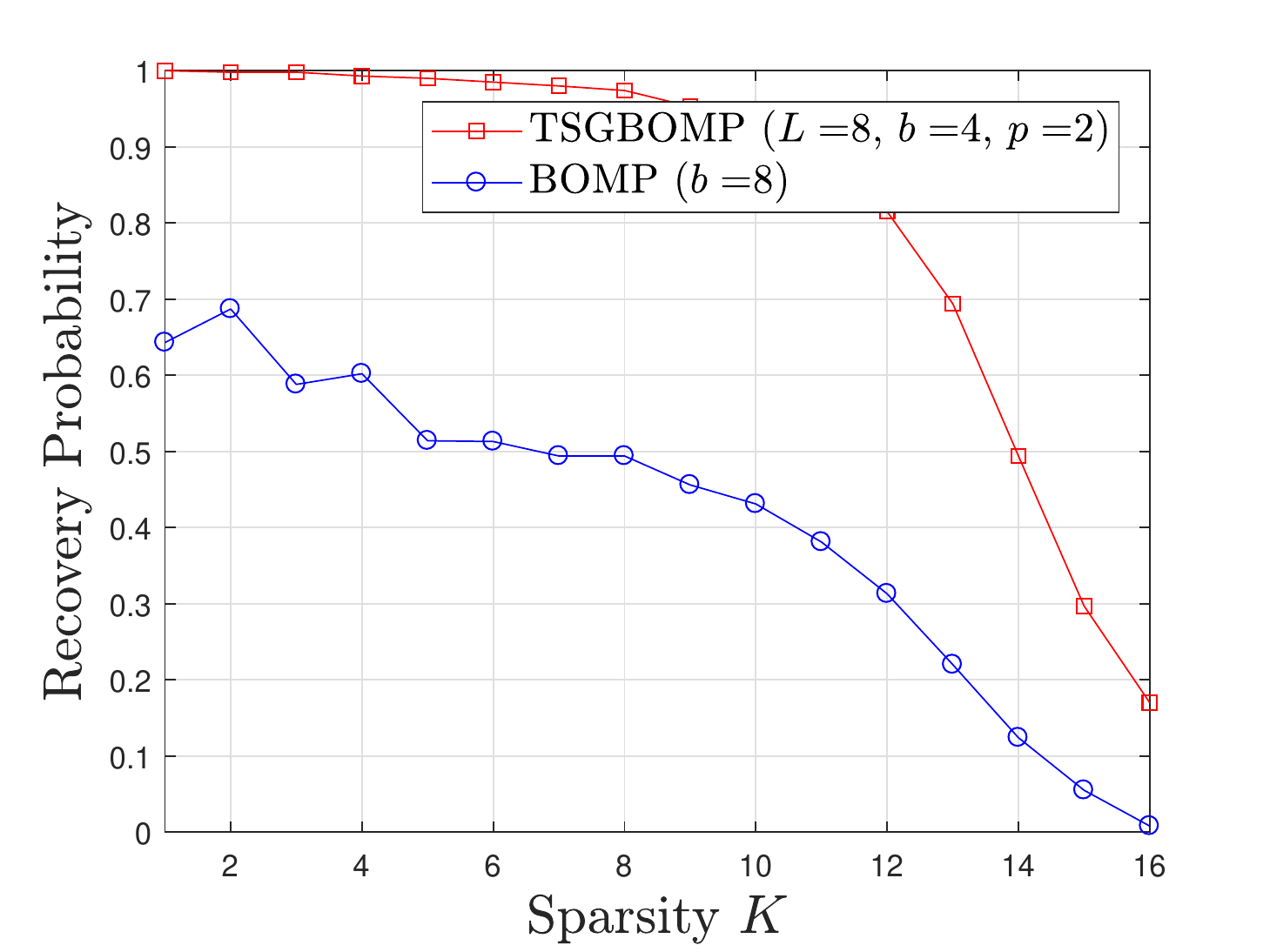}
		\caption{$p=2$}
		\label{fig:prob-recovery-vs-sparsity-m=120-n=200-p=2}
	\end{subfigure}
	\caption{Probability of recovery vs sparsity for $n=200,\ m=120$}
	\label{fig:prob-recovery-vs-sparsity-m=120-n=200}
\end{figure}
\begin{figure}[t!]
	\centering
	\begin{subfigure}{0.5\textwidth}
		\centering
		\includegraphics[height=1.5in,width=3in]{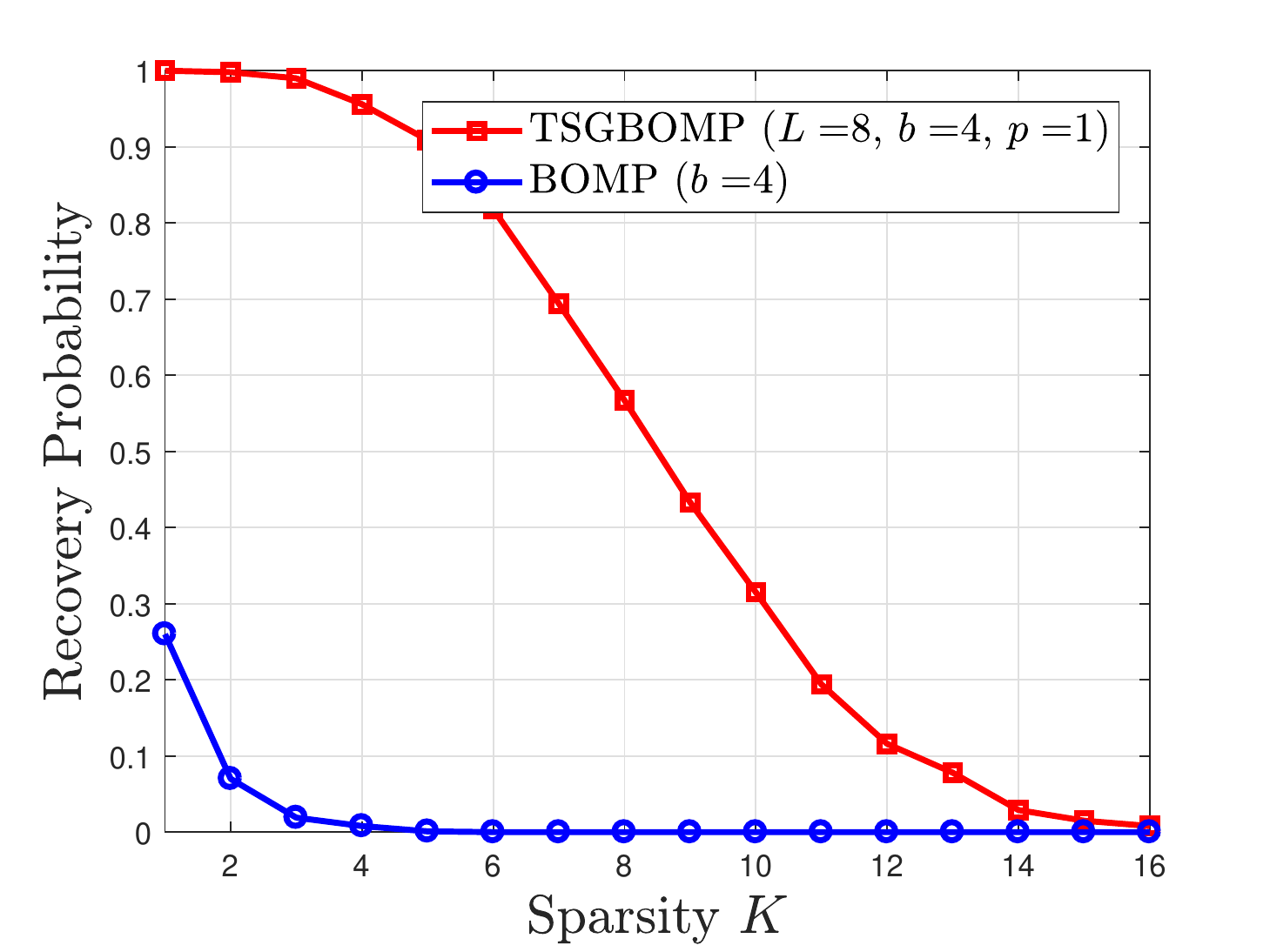}
		\caption{$p=1$}
		\label{fig:prob-recovery-vs-sparsity-m=160-n=200-p=1}
	\end{subfigure}
	\begin{subfigure}{0.5\textwidth}
		\centering
		\includegraphics[height=1.5in,width=3in]{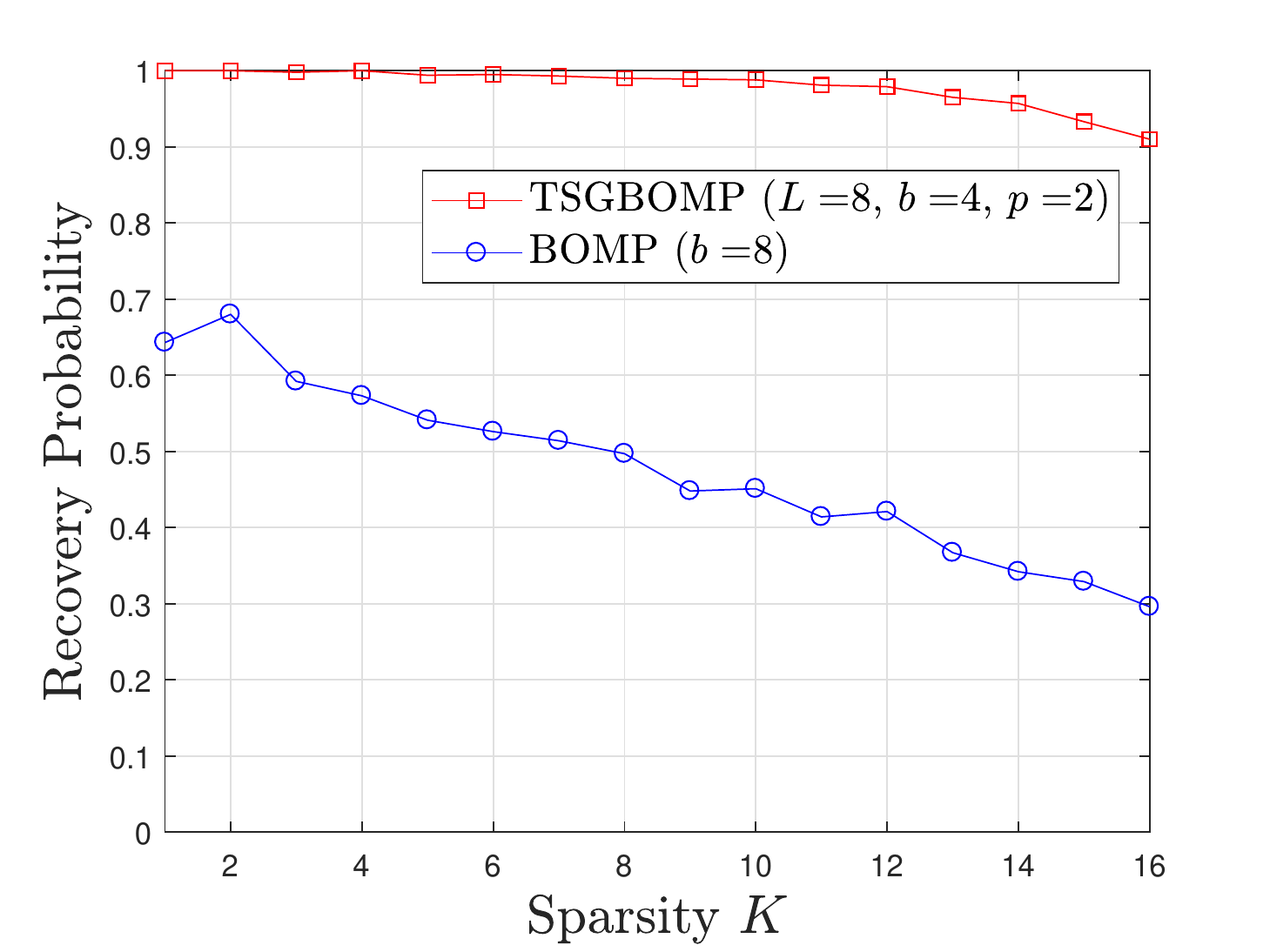}
		\caption{$p=2$}
		\label{fig:prob-recovery-vs-sparsity-m=160-n=200-p=2}
	\end{subfigure}
	\caption{Probability of recovery vs sparsity for $n=200,\ m=160$}
	\label{fig:prob-recovery-vs-sparsity-m=160-n=200}
\end{figure}
%
%
In this section we compare the recovery success performance of TSGBOMP with BOMP, over a range of sparsity values for different sensing matrix sizes. In all the simulation exercises, we generate random matrices with size $m\times n$, with i.i.d. Gaussian entries with zero mean and unit variance, and the columns are normalized to have unit norm. For the unknown signal, we generate a PIBS vector with parameters $(b,p,L,L',K,0)$. Here the parameters $n,b,p,L,K$ are chosen such that the signal is large enough to accommodate the desired signal structure and $1\le p\le K$. Once the support of such a vector is generated, the entries are filled randomly with values $\pm 10$. The results are averaged over $1000$ such trials. All the experiments are carried out on MATLAB 2017a running on a Core i-5 Laptop with 8 GB RAM, 64 bit processor and 1.60 GHz processor speed.

We have conducted two sets of experiments with $n=200$, and $m=120$ and $160$. In both of these experiments, we have fixed $b=4$, and have taken $p=1,2$. The parameter $L$ is always taken to be $L=8$. The sparsity $K$ is varied over the range $1$ to $16$. In these experiments, the TSGBOMP algorithm is executed with parameters $K,B=bp,L\ge bp$; and the algorithms GBOMP and BOMP are executed with parameter $bp$. In all these experiments, we see from the figures~\ref{fig:prob-recovery-vs-sparsity-m=120-n=200} and~\ref{fig:prob-recovery-vs-sparsity-m=160-n=200} that while both the algorithms GBOMP and TSGBOMP exhibit moderately well recovery of probability performance, the BOMP algorithm has abysmal performance in recovering the unknown vector. This is to be expected, as the BOMP algorithm was not designed to consider the recovery of SBS type vectors which might not exactly fit into prespecified block boundaries, and might have partial overlaps in consecutive blocks. Moreover, we see from figures~\ref{fig:prob-recovery-vs-sparsity-m=120-n=200-p=1},~\ref{fig:prob-recovery-vs-sparsity-m=120-n=200-p=2},or from figures~\ref{fig:prob-recovery-vs-sparsity-m=160-n=200-p=1},~\ref{fig:prob-recovery-vs-sparsity-m=160-n=200-p=2}, that the recovery performances of TSGBOMP and GBOMP improve significantly when $p$ is increased. This improvement can be explained by the fact that increasing $p$ produces PIBS vectors with larger nonzero (true) blocks with size varying from $b$ to $bp$, while keeping the total size of the sum of the nonzero blocks to $Kb$. This implies that once a true block is found, more than one blocks of size $b$ are found, which reduces the chance of detection of false blocks, thereby increasing the chance of signal recovery.

%

%
\appendices
%
\section{Proofs of the Lemmas in Section~\ref{sec:useful-definitions-lemmas} }
\subsection{Proof of Lemma~\ref{lem:rip-PIBS-eigenvalue-connection}}
\label{sec:appendix-proof-rip-eigenvalue-condition}
From the Definition~\ref{def:RIP}, $\delta\ge \delta_{b,p,l,L'}(K,R)$, for every PIBS vector $\bvec{x}\in \complex^{n}$ with parameters $(b,p,l,L',k,r),\;0\le k\le K,\;0\le r \le R$. If $S$ be the support of the vector ($S\in \cup_{r=0}^R\cup_{k=0}^K\Sigma_{b,p,l,L'}(k,r)$), then, $\delta_{b,p,l,L'}(K,R)\norm{\bvec{x}}^2=\delta_{b,p,l,L'}(K,R)\norm{\bvec{x}_S}^2
\ge \abs{{\bvec{x}_S}^H(\bvec{\Phi}_S^H\bvec{\Phi}_S-\bvec{I}_{|S|\times |S|}){\bvec{x}_S}}=\abs{\norm{\bvec{\Phi x} }^2-\norm{\bvec{x}}^2}$, which follows from the definition of the operator norm .
From this, the inequality~\eqref{eq:PIBS-rip-inequalities} follows trivially for any $\delta\ge \delta_{b,p,l,L'}(K,R)$.
\subsection{Proof of Lemma~\ref{lem:PIBRIP-K-R-monotonicity}}
\label{sec:appendix-proof-pibrip-k-r-monotonicity}
The proof follows immediately after writing down the expressions for $\delta_{b,p,l,L'}(K_i,R_j),\ 1\le i,j\le 2$ using Eq~\eqref{eq:PIBS-rip-eigenvalue-connection} from definition~\ref{def:RIP} and subsequently noting that if $K_i(\mathrm{resp.}\ R_i)\le K_j(\mathrm{resp.}\ R_j)$ then the collection of sets over which the search for $\delta_{b,p,l,L'}(K_j,R_j)$ is conducted is a superset of the collection of sets over which $\delta_{b,p,l,L'}(K_i,R_i)$ is searched.
\subsection{Proof of Lemma~\ref{lem:PIBRIP-L-monotonicity}}
\label{sec:appendix-proof-pibrip-l-monotonicity}
Choose some $k,r$ such that $0\le k\le K,\ 0\le r\le R$, $R=0,1,2$. Let $S\subset [n]$ be an arbitrary set $S\in \Sigma_{b,p,l,L'}(k,r)$. Then, the Lemma~\ref{lem:PIBRIP-L-monotonicity} will be proved if we can prove that $\opnorm{\bvec{\Phi}_S^H\bvec{\Phi}_S-\bvec{I}_S}{2\to 2}\le \delta_{b,p,L',L'}(K,R)$, i.e., if we can show that there exists a set $S'\subset [n]$ such that $S'\in \cup_{k=0}^K\cup_{r=0}^R\Sigma_{b,p,L',L'}(k,r)$ and $\opnorm{\bvec{\Phi}_S^H\bvec{\Phi}_S-\bvec{I}_S}{2\to 2}\le\opnorm{\bvec{\Phi}_{S'}^H\bvec{\Phi}_{S'}-\bvec{I}_{S'}}{2\to 2}$, which in turn, will be proved if one can show that $S\subset S'$ for some $S'\in \Sigma_{b,p,L',L'}(k',r')$ for some $0\le k'\le K$, and $0\le r'\le R$.

If $R=0$, we must have $r=0$. Then, the set $S$ corresponds to indices of true clusters only, in which case one can simply take $S'=S$ to prove the above claim.

If $R=1$, one can have $r=0,1$. The case $r=0$ is already considered. For $r=1$, given $S\in \Sigma_{b,p,l,L'}(k,r)$, it is always possible to construct at least one $S'\in \Sigma_{b,p,L',L'}(k',r')$ with $r'=1$,
which covers $S$, i.e., $S\subset S'$. To show this, we define the \emph{leading signal edge} and the \emph{trailing signal edge} respectively as the subset of $[n]$ lying before the first and after the last true cluster. Let the corresponding lengths be $l_1$ and $l_2$ respectively. Then, for $r=1$, if the pseudo block lies between two true clusters, or, in the leading (trailing) signal edge with $l_1\ge L'$ ($l_2\ge L'$), it can be covered by a larger pseudo block of size $L'$ without disturbing the true clusters.
If, on the other hand, it lies in the leading signal edge with $l_1 < L'$, one has to move the first true cluster of $S$, of size, say, $f_1$ to the front (i.e., from index $1$ to $f_1$), followed by a larger pseudo block of size $L'$ while retaining other true clusters of $S$ (if it lies in the trailing signal edge with $l_2 < L'$, similar treatment in the opposite direction may be used). It is easy to see that the $S'$ created in all the above cases contains $S$.

For $R=2$, one can have $r=0,1,2$. It is sufficient to consider the case $r=2$ as the constructions for $S'$ for $r=0,1$ have already been described before. For $r=2$, there are two nontrivial cases (other cases are straightforward and can be handled by direct application of the above arguments).
First, consider the case where one pseudo block lies in one signal edge with $l_1< L'$ while the other one lies in the immediate next inter-cluster space. In this case, as above, we move the first true cluster of $S$ to the front (i.e., from index $1$ to $f_1$), followed by a larger pseudo block of size $L'$. If the starting index of the second pseudo block of $S$, say, $n_1$ ($n_1\ge l_1+f_1+1$) is such that
$n_1\le L'+f_1$, then the first $L'+f_1$ indices have a non-empty overlap with the second pseudo block in $S$. In such case, we place the large pseudo block of size $L'$ from the $(L'+f_1+1)^\mathrm{st}$ index, else, it is placed from index $n_1$. It is easy to see that these $L'$ indices cover the second pseudo block of $S$. In both cases, if these $L'$ indices do not overlap with the second true cluster in $S$, then the first $L'+f_1$ indices along with these $L'$ indices cover the two pseudo blocks as well as the first true cluster in $S$, and so these indices, along with the support of the other clusters of $S$ generate the desired $S'$. Otherwise, instead of using $L'$ indices to place the second large pseudo block, we consider $L'+f_2$ indices where the first $f_2$ indices support the second true cluster and the subsequent $L'$ indices support the second large pseudo block. Clearly, this results in a desired $S'$ that covers $S$. The other nontrivial case for $r=2$ occurs when both the pseudo blocks in $S$ lie in the same inter-cluster space. Here, we start from the first index of the first pseudo block of $S$ and cover it by $L'$ indices, placing the first large pseudo block there. If this does not overlap with the second pseudo block, we cover the second pseudo block with another set of $L'$ indices starting from its first index and place the second pseudo block there, else we start from right after the first set of $L'$ indices. Again, if the second pseudo block overlaps with the next true cluster (of size, say, $f_3$), instead of $L'$ indices, we consider a set of $L'+f_3$ indices to support first the true cluster and then the second pseudo block of size $L'$ (in case the first set of $L'$ indices itself overlaps with the next true cluster, we readjust it so that it covers both the pseudo blocks and remains confined in the inter cluster space). Clearly, in all the above cases, a desired support set $S'$ is generated that contains $S$. This proves the Lemma.

\subsection{Proof of Lemma~\ref{lem:PIBRIP-B-L-monotonicity} }
\label{sec:appendix-proof-pibrip-b-l-monotonicity}
Consider a PIBS vector with parameters $(b,p,L',L',K,1),\;K\ge 1$ and let $S$ denote its support set. The solitary pseudo block may lie either in between two consecutive true clusters or beside just one true cluster. Since we are dealing with both
$\delta_{b,p,L',L'}(K,1)$ as well $\delta_{b,p,L',L'}(K-1,2)$, it is implicit that the size of the index set $n$ is large enough to accommodate both $K$ true blocks together with one pseudo block of size $L'$, as well as $K-1$ true blocks together with two pseudo blocks of size $L'$ each. It then easily follows that $b$ consecutive indices (i.e., one true block) either from the front of the true cluster to the left or from the rear of the true cluster to the right can be covered by part of a pseudo block of size $L'$, with the remaining part of the pseudo block of size $(L'-b)$ supported in the inter-cluster space. This results in a support set $S'$ with $S\subset S'$ and thus, $\delta_{b,p,L',L'}(K,1)\le \delta_{b,p,L',L'}(K-1,2)$.
%
\subsection{Proof of Lemma~\ref{lem:projected-matrix-pibrip} }
\label{sec:appendix-proof-projected-matrix-pibrip}
The proof extends the arguments of the proof of Lemma~$5$ of~\cite{cai2011orthogonal} with modifications to the PIBS structure.
Since $\dualproj{S_1}\bvec{\Phi}_{S_2}\bvec{x}_{S_2}=\dualproj{S_1}\bvec{\Phi}_{S_2\setminus S_1}\bvec{x}_{S_2\setminus S_1'}$, one can partition the matrix $\bvec{\Phi}_S$ as $\bvec{\Phi}_S = [\bvec{\Phi}_{S_2\setminus S_1'}\ \bvec{\Phi}\curly{S_1}]$ where, for brevity, we have written $S_1'=I(\bvec{\Phi}\curly{S_1})$. This representation of $\bvec{\Phi}_S$ will then be used, as in~\cite{cai2011orthogonal}, to determine that the minimum eigenvalue of $\left(\bvec{\Phi}_{S_2\setminus S_1'}^H\dualproj{S_1}\bvec{\Phi}_{S_2\setminus S_1'}\right)^{-1}$ is lower bounded by the minimum eigenvalue of $\left(\bvec{\Phi}_S^H\bvec{\Phi}_S\right)^{-1}$, and the maximum eigenvalue of $\left(\bvec{\Phi}_{S_2\setminus S_1'}^H\dualproj{S_1}\bvec{\Phi}_{S_2\setminus S_1'}\right)^{-1}$ is upper bounded by the maximum eigenvalue of $\left(\bvec{\Phi}_S^H\bvec{\Phi}_S\right)^{-1}$. By definition~\ref{def:RIP}, $\frac{1}{1+\delta_{b,p,l,L'}(K,R)}\le \lambda_{\min}\left(\left(\bvec{\Phi}_S^H\bvec{\Phi}_S\right)^{-1}\right)\le \lambda_{\max}\left(\left(\bvec{\Phi}_S^H\bvec{\Phi}_S\right)^{-1}\right)\le \frac{1}{1-\delta_{b,p,l,L'}(K,R)}$. Using definition~\ref{def:RIP} once again for the matrix $\dualproj{S_1}\bvec{\Phi}_{S_2\setminus S_1'}$, one concludes that $\delta_{b,p,l,L'}(K,R)\ge \opnorm{\bvec{\Phi}_{S_2\setminus S_1'}^H\dualproj{S_1}\bvec{\Phi}_{S_2\setminus S_1'}-\bvec{I}_{\abs{S_2\setminus S_1'}\times \abs{S_2\setminus S_1'}}}{2\to 2}$ (Here we have used the fact that $(\dualproj{S_1})^H\dualproj{S_1}=\dualproj{S_1}$). Then, using Lemma~\ref{lem:rip-PIBS-eigenvalue-connection} for the matrix $\dualproj{S_1}\bvec{\Phi}_{S_2\setminus S_1'}$, the result follows.
\subsection{Proof of Lemma~\ref{lem:projected-matrix-innerproduct-inequality}}
\label{sec:appendix-proof-projected-matrix-innerproduct-inequality}
Note that the vectors $\bvec{u}$ and $\bvec{v}$ are orthogonal since they are supported on $S_2$ and $S_3$ respectively, which are disjoint. Therefore, writing $S'=S_2\cup S_3$, one obtains, \begin{align}
\lefteqn{\abs{\inprod{\dualproj{S_1}\bvec{\Phi u}}{\bvec{\Phi v}}}} & &\nonumber\\
\ & = \abs{\inprod{\bvec{u}_{S'}}{(\bvec{\Phi}_{S'}^H\dualproj{S_1}\bvec{\Phi}_{S'}-\bvec{I}_{\abs{S'}\times \abs{S'}})\bvec{v}_{S'}}}\nonumber\\
\ & \le \opnorm{\bvec{\Phi}_{S'}^H\dualproj{S_1}\bvec{\Phi}_{S'}-\bvec{I}_{\abs{S'}\times \abs{S'}}}{2\to 2}\norm{\bvec{u}_{S'}}\norm{\bvec{v}_{S'}}\nonumber\\
\ & \le \delta_{b,p,l,L'}(K,R)\norm{\bvec{u}}\norm{\bvec{v}},
\end{align}
where the last step uses Lemma~\ref{lem:projected-matrix-pibrip} since $I(\bvec{\Phi}\curly{S_1})\cup S'$ is the support of a PIBS vector with parameters $(b,p,l,L',k,r),\;0\le k\le K,\;0\le r\le R$.
\subsection{Proof of Lemma~\ref{lem:projected-column-lower-bound-wielandt-PIBS}}
\label{sec:appendix-proof-projected-column-lower-bound-wielandt-pibs}
A similar result has been proved in~\cite{wen2017nearly} in the context of sparse vectors. We prove this results here using Wielandt's theorem~\cite[Theorem $7.4.34$]{horn1990matrix}, which states that for any two vectors $\bvec{u},\bvec{v}\in\complex^n$, that are orthogonal, i.e., $\inprod{\bvec{u}}{\bvec{v}}=0$, and for any Hermitian positive definite matrix $\bvec{B}\in \complex^{n\times n}$, one has \begin{align}
\label{eq:wielandt-theorem-statement}
\abs{\bvec{u}^H\bvec{B}\bvec{v}}^2\le \left(\frac{\lambda_{\max}(\bvec{B}) - \lambda_{\min}(\bvec{B})}{\lambda_{\max}(\bvec{B}) + \lambda_{\min}(\bvec{B})}\right)^2 (\bvec{u}^H\bvec{B}\bvec{u})(\bvec{v}^H\bvec{B}\bvec{v}).
\end{align}
Now, $\norm{\dualproj{S}\bvec{\phi}_j}^2 = 1 - \norm{\proj{S}\bvec{\phi}_j}^2$. One can express $\proj{S}\bvec{\phi}_j$ as $\bvec{\Phi}\curly{S}\bvec{z}\curly{S}$, for some vector $\bvec{z}\in \complex^{n}$ supported on $S$. Define $S'=I(\bvec{\Phi}\curly{S})\cup \{j\}$. Then, $\norm{\proj{S}\bvec{\phi}_j}^2 = \inprod{\proj{S}\bvec{\phi}_j}{\proj{S}\bvec{\phi}_j}= \inprod{\proj{S}\bvec{\phi}_j}{\bvec{\phi}_j}= 
\inprod{\bvec{\Phi}_{S'}\bvec{z}_{S'}}{\bvec{\Phi}_{S'}{\bvec{e}}_{j,S'}}$, where $\bvec{e}_j\in \complex^{n}$, with $[\bvec{e}_j]_j=1$, and $[\bvec{e}_j]_l=0$ for all $l\ne j$, $l\in [n]$. Note that $\inprod{\bvec{z}}{\bvec{e}_j}=0$, since $j\notin I(\bvec{\Phi}\curly{S})$. By assumption, $S'$ can be associated with a PIBS vector with parameters $(b,p,1,L',k,1),\;0\le k\le K$. Since $\delta_{b,p,1,L'}(K,1)<1$, we have,
\begin{align}
1-\delta_{b,p,1,L'}(K,1) & \le \lambda_{\min}\left(\bvec{\Phi}_{S'}^H\bvec{\Phi}_{S'}\right) \nonumber\\ \le \lambda_{\max}\left(\bvec{\Phi}_{S'}^H\bvec{\Phi}_{S'}\right) & \le 1 + \delta_{b,p,1,L'}(K,1).
\end{align} Now we use the Wielandt's theorem where in ~\eqref{eq:wielandt-theorem-statement}, we replace $\bvec{B}$, $\bvec{u}$
and $\bvec{u}$ respectively by $\bvec{\Phi}_{S'}^H\bvec{\Phi}_{S'}$, $\bvec{z}_{S'}$ and $\bvec{e}_{j,S'}$. Defining $\lambda=
\frac{\lambda_{\max}(\bvec{B})}{\lambda_{\min}(\bvec{B})}$ and $\delta=\frac{1 + \delta_{b,p,1,L'}(K,1)}{1 - \delta_{b,p,1,L'}(K,1)}$, from above, we have, $\lambda \le \delta$. From this and the fact that $f(\lambda)=\frac{\lambda-1}{\lambda+1}$ is an increasing function in $\lambda$, from ~\eqref{eq:wielandt-theorem-statement},
we have,
$\norm{\proj{S}\bvec{\phi}_j}^2\le \delta_{b,p,1,L'}(K,1)\norm{\bvec{\Phi}_{S'}\bvec{z}_{S'}}\norm{\bvec{\phi}_j}=\delta_{b,p,1,L'}(K,1)\norm{\proj{S}\bvec{\phi}_j}\implies\norm{\proj{S}\bvec{\phi}_j}\le \delta_{b,p,1,L'}(K,1) $, where we have used the fact that  $\norm{\bvec{\phi}_j}=1$, and $\bvec{\Phi}_{S'}\bvec{z}=\proj{S}\bvec{\phi}_j$. Hence~\eqref{eq:projected-column-lower-bound-wielandt-pibrip} follows.
\section{Proof of Lemma~\ref{lem:lower-bound-real-part}}
\label{sec:appendix-proof-lemma-lower-bound-real-part}
Let $v=w/z=Re^{j\theta}$. Then, \begin{align}
\Re\left(\frac{(z+w)}{\abs{z+w}}z^*\right) & = \abs{z}\Re\left(\frac{1+v}{\abs{1+v}}\right)\nonumber\\
\ & = \frac{\abs{z}(1+R\cos\theta)}{\sqrt{1+R^2+2R\cos \theta}}.
\end{align}
Now, if $v$ is real, $\theta=\pi l$ for some integer $l$, so that, for real $v$, using the fact that $R\in [0,1]$, \begin{align}
\frac{1+R\cos\theta}{\sqrt{1+R^2+2R\cos \theta}} & = \frac{1\pm R}{\abs{1\pm R}} = 1.
\end{align}
On the other hand, for general $v$, for a fixed $R\in [0,1]$, we consider the function \begin{align}
f(\theta) & = \frac{1+R\cos \theta}{\sqrt{1+R^2+2R\cos \theta}}.
\end{align}
Note that \begin{align}
f(\theta) & = \frac{1}{\sqrt{2}}\left[\sqrt{R\cos\theta + \frac{R^2+1}{2}}+\frac{\left(\frac{1-R^2}{2}\right)}{\sqrt{R\cos\theta + \frac{R^2+1}{2}}}\right]\nonumber\\
\ & = \frac{1}{\sqrt{2}}\tau\left(\sqrt{R\cos\theta + \frac{R^2+1}{2}}\right),
\end{align}
where $\tau(x)=x+\frac{\left(\frac{1-R^2}{2}\right)}{x}$. Therefore, finding the minimum non-negative value of $f$ is equivalent to finding the minimum non-negative value of $\tau(x)$ for $x=\sqrt{R\cos\theta + \frac{R^2+1}{2}}$. Therefore, as $R^2\le 1$, the minimum non-negative value of $f$ is obtained when $R\cos\theta +\frac{R^2+1}{2}=\frac{1-R^2}{2}$, i.e., when $\cos\theta =-R$. Therefore, \begin{align}
f(\theta)\ge \frac{1-R^2}{\sqrt{1+R^2-2R^2}}=\sqrt{1-R^2}.
\end{align}
\section{Proof of Theorem~\ref{thm:random-matrix-guarantees} }
\label{sec:appendix-proof-thm-random-martix-guarantees}
Let $F$ be the distribution function of $\delta_{b,p,L',L'}(K-1,2)$. For simplicity of writing, in the sequel, we will denote $\delta_{b,p,L',L'}(K-1,2)$ by $\delta$.

Let $\zeta$ denotes the probability density function associated to the random variable $\frac{x_{\min}}{x_{\max}}$. Then, using the independence of the random variables $\frac{x_{\min}}{x_{\max}}$ and $\delta$, one can write, \begin{align}
\lefteqn{\prob{x_{\min}>f_K(\delta)x_{\max}, \delta<\frac{1}{\sqrt{2K+1}}}} & &\nonumber\\
\ & \ge  1 - \int_0^1 \prob{\delta\ge f_K^{-1}\left(r\right)}\zeta(r)dr-\prob{\delta\ge \frac{1}{2K+1}}\nonumber\\
\label{eq:prob-identity-as-integral-first-step}
\ & = 1 - \int_0^{f_K^{-1}\left(1\right)}\prob{\delta\ge \delta_0}\zeta(f_K(\delta_0))f_K'(\delta_0)d\delta_0-\prob{\delta\ge \frac{1}{2K+1}},
\end{align}
where we have used the fact that $f_K^{-1}(0)=0.$
Now, we proceed to find a lower bound of the above expression by finding an upper bound on the complementary CDF of $\delta$, that is, we find an upper bound of the tail probability $\prob{\delta>\delta_0}$. In order to do so, we adopt the approach used in the proof of Proposition $3$ of~\cite{eldar2009robust}. However, our analysis will involve computing a quite different combinatorial quantity that expresses the number of ways the support of a PIBS vector can be generated.

Now recalling the definition~\ref{def:RIP} and Lemma~\ref{lem:PIBRIP-K-R-monotonicity}, it can be observed that, for any $b,p\ge 1,l,k,r\ge 0$, whenever $k\le k'$ and $r\le r'$, for any subset $S\in \Sigma_{b,p,l,l}(k,r)$, there is a subset $S'\in\Sigma_{b,p,l,l}(k',r')$ such that $S\subseteq S'$. Hence, $\delta=\max_{S\in \Sigma_{b,p,L',L'}(K-1,2)}\opnorm{\bvec{\Phi}_S^H\bvec{\Phi}_S-\bvec{I}_S}{2\to 2}=\max\{\overline{\sigma}^2-1,1-\underline{\sigma}^2\}$ (since the matrix is constructed with real Gaussian entries), where $\overline{\sigma}=\max_{S\in \Sigma_{b,p,L,L}(K-1,2)}\sigma_{\max}(\bvec{\Phi}_S)$, and $\underline{\sigma}=\min_{S\in \Sigma_{b,p,L,L}(K-1,2)}\sigma_{\min}(\bvec{\Phi}_S)$, where $\sigma_{\max}(\bvec{\Phi}_S)$ is the maximum singular value of $\bvec{\Phi}_S$, and $\sigma_{\min}(\bvec{\Phi}_S)$ is the minimum (positive) singular value of $\bvec{\Phi}_S$. Then, it follows that \begin{align}
\prob{\delta>\delta_0} & = \prob{\overline{\sigma}>\sqrt{1+\delta_0}} + \prob{\underline{\sigma}^2<1-\delta_0}.
\end{align}
Let $\gamma=\sqrt{1+\delta_0}-1$. Then, note that $\underline{\sigma}^2< 1-\delta_0$ implies that $\underline{\sigma}< 1-\gamma$, since $(1-\gamma)^2-(1-\delta_0)=(2-\sqrt{1+\delta_0})^2-(1-\delta_0)=4-4\sqrt{1+\delta_0}+2\delta_0=2(\sqrt{1+\delta_0}-1)^2>0$. Therefore, \begin{align}
\label{eq:prob-upper-bound-preliminary}
\prob{\delta> \delta_0} & \le \prob{\overline{\sigma}>1+\gamma}+\prob{\underline{\sigma}<1-\gamma}.
\end{align}
Clearly, we need $\gamma<1$, which is ensured if $\delta_0<3$.
Now using the result from Davidson and Szarek~\cite{szarek1991condition,davidson2001local} along with the
analysis that produced Eq~$(78)$ in~\cite{eldar2009robust}, it follows that, \begin{align}
\prob{\overline{\sigma}>1+\lambda+t} & \le N_{b,p,L'}(K-1,2)e^{-mt^2},\ (t>0).
\end{align}
where $\lambda=\sqrt{((K-1)b+2L')/m}$ and $N_{b,p,L'}(K-1,2)$ is the number of ways of choosing supports for a PIBS vector with parameters $(b,p,L',L',K-1,2)$. Here it is assumed that $(K-1)b+2L'\le m$. Similarly, it can be derived that \begin{align}
\prob{\underline{\sigma}<1-\lambda-t}& \le N_{b,p,L'}(K-1,2)e^{-mt^2},\ (t>0).
\end{align}
Thus, from inequality~\eqref{eq:prob-upper-bound-preliminary}: \begin{align}
\prob{\delta>\delta_0} & \le 2N_{b,p,L'}(K-1,2)e^{-mt^2},
\end{align}
whenever $\lambda+t = \gamma$, or equivalently, $t = \sqrt{1+\delta_0}-(1+\lambda)$. Note that $t>0$ if \begin{align}
1+\lambda < \sqrt{1+\delta_0}\Leftrightarrow \nu & <\delta_0,
\end{align}
where $\nu=\lambda^2+2\lambda$. Furthermore, as $\delta_0<f_K^{-1}(1)$, if we choose $b,p,L,K,m$ such that $\nu<\rho<3$, where $\rho=f_K^{-1}(1)$, we obtain, $\delta_0\in [\nu,\rho]$. Therefore, as long as $\delta_0\in[\nu,\rho]$, the following holds: \begin{align}
\label{eq:tail-probability-PIBRIC}
\prob{\delta>\delta_0} & \le 2N_{b,p,L'}(K-1,2)e^{-m\left[\sqrt{1+\delta_0}-(1+\lambda)\right]^2}.
\end{align}
We now find an upper bound on $N_{b,p,L'}(K-1,2)$. In Appendix~\ref{sec:appendix-upper-bound-combinatorial-N} we have established the following upper bound on $N_{b,p,L'}(K,R)$:
\begin{prop}
	\label{prop:upper-bound-combinatorial-N}
	If $L\ge pb, (R+1)p\le K$, then \begin{align}
	\label{eq:upper-bound-combinatorial-N}
	N_{b,p,L'}(K,R) & < \exp\left[A(p,K,R)+ KC(p)\right.\nonumber\\
	\ & \left.+K\ln\left(pD(n,b,K,L')/K-E(L')\right)\right],
	\end{align}
	where $A(p,K,R)=\frac{3pK}{2(p+1)^2}+R,\ C(p) = \ln p + 21/ 8 -1/p,\ D(n,b,K,L') = n-Kb+L',\ E (L')= L'-1$.
\end{prop}
Hence, assuming that $3p+1\le K$, we have $N_{b,p,L'}(K-1,2)< \exp\left(A+KC+K\ln(pD/K-E)\right)=e^{h(b,p,L',K,R)}$, where $h(b,p,L',K,R)=A+KC+K\ln(pD/K-E),$ and $A=\frac{3p(K-1)}{2(p+1)^2}+1,\ C=\ln p + 21/ 8 -1/p,\ D=n-(K-1)b+L',\ E=L'-1$.\\
Now, note that, since we have assumed that $\delta_0\le \rho$, we have, for any $\delta_0\in [\nu,\rho]$, \begin{align}
\sqrt{1+\delta_0}-(1+\lambda) & \ge \frac{\delta_0-\nu}{1+\lambda+\sqrt{1+\rho}}.
\end{align}
Therefore, from Proposition~\ref{prop:upper-bound-combinatorial-N}, one obtains, for any $\delta_0\in [\nu,\rho]$,\begin{align}
\prob{\delta>\delta_0} & \le 2e^{h(b,p,L',K,R)}e^{-\frac{m(\delta_0-\nu)^2}{(1+\lambda+\sqrt{1+\rho})^2}}\nonumber\\
\ & = c_1e^{-c_2(\delta_0-\nu)^2},
\end{align} 
where $c_1 = 2 e^{h(b,p,L',K,R)},\ c_2 = \frac{m}{(1+\lambda+\sqrt{1+\rho})^2}$.

Consequently, assuming that $K,m$ are chosen in such a way that there is $\epsilon_0>0$ such that \begin{align}
\label{eq:prob-delta-K-condition}
\nu + \epsilon_0<\frac{1}{\sqrt{2K+1}} & < \rho, 
\end{align}
we obtain, \begin{align}
\prob{\delta\ge \frac{1}{\sqrt{2K+1}}} & \le c_1e^{-c_2\epsilon_0^2}.
\end{align}

Furthermore, from~\eqref{eq:prob-identity-as-integral-first-step}, we obtain, for any $\epsilon\in [0,\rho-\nu]$ \begin{align}
\lefteqn{\prob{x_{\min}>f_K(\delta)x_{\max}}} & &\nonumber\\
\ & > 1 - \int_0^{\nu+\epsilon}1\cdot \zeta(f_K(\delta_0))f_K'(\delta_0)d\delta_0\nonumber\\
\ & -\int_{\nu+\epsilon}^{\rho}c_1e^{-c_2(\delta_0-\nu)^2}\zeta(f_K(\delta_0))f_K'(\delta_0)d\delta_0\nonumber\\
\ & > 1 - \prob{\frac{x_{\min}}{x_{\max}}\le f_K(\nu+\epsilon)}\nonumber\\
\ & - c_1e^{-c_2\epsilon^2}\prob{f_K(\nu+\epsilon)<\frac{x_{\min}}{x_{\max}}\le 1}\nonumber\\
\label{eq:probability-lower-bound-preliminary}
\ & = g\left(f_K(\nu+\epsilon)\right)\left(1-c_1e^{-c_2\epsilon^2}\right),
\end{align}
where $g$  is the complementary CDF of the random variable $\frac{x_{\min}}{x_{\max}}$, i.e., for any $a\in [0,1]$, $g(a)=\prob{\frac{x_{\min}}{x_{\max}}>a}$. Therefore, from Eqs.~\eqref{eq:prob-delta-K-condition} and~\eqref{eq:probability-lower-bound-preliminary} we obtain the bound~\eqref{eq:probability-lower-bound}.

We now find upper and lower bounds of $g$ as below. In order to do so, first observe that the function $g$ can be expressed as below: 
\begin{align}
g(a) & = \sum_{i=1}^{Kb} \prob{\cap_{j\ne i} \left\{\abs{x_i}\le \abs{x_j}<\frac{\abs{x_i}}{a}\right\}}\nonumber\\
\ & \stackrel{(\psi_1)}{=} Kb\prob{\cap_{j>1}\left\{\abs{x_1}\le \abs{x_j}<\frac{\abs{x_1}}{a}\right\}}\nonumber\\
\label{eq:g-function-expression}
\ & \stackrel{(\psi_2)}{=} Kb\int_0^\infty \left(\prob{z\le X<\frac{z}{a}}\right)^{Kb-1}dF_Z(z).
\end{align}
Here, the steps $(\psi_1)$ and $(\psi_2)$ follow from the fact that the random variables $x_1,\cdots,x_{Kb}$ are i.i.d. distributed as the random variable $X,Z$, which are i.i.d. and have half normal distribution, i.e., the density of $X,Z$ is given by $p(u)=\sqrt{\frac{2}{\pi}}e^{-u^2/2},\ u\ge 0$. Therefore, from~\eqref{eq:g-function-expression} we obtain that,
\begin{align}
\lefteqn{g(a)} & &\nonumber\\
\ & = Kb\expect{Z}{\left(\expect{X}{\indicator{Z<X<\frac{Z}{a}}}\right)^{Kb-1}}\nonumber\\
\ & \stackrel{(\psi_3)}{\ge} Kb\left(\expect{X}{\expect{Z}{\indicator{Z<X<\frac{Z}{a}}}}\right)^{Kb-1}\nonumber\\
\ & \stackrel{(\psi_4)}{=} Kb\left(\prob{1<F<\frac{1}{a^2}}\right)^{Kb-1}\nonumber\\
\ & \stackrel{(\psi_5)}{=} Kb\left(I_{\frac{1}{1+a^2}}\left(\frac{1}{2},\frac{1}{2}\right) - I_{\frac{1}{2}}\left(\frac{1}{2},\frac{1}{2}\right)\right)^{Kb-1}\nonumber\\
\label{eq:g-function-lower-bound}
\ & \stackrel{(\psi_6)}{=} Kb\left(\frac{2}{\pi}\cot^{-1}a - \frac{1}{2}\right)^{Kb-1}.
\end{align}
Here, the step $(\psi_3)$ follows from Jensen's inequality, since, when $x\ge 0$, the function $x^r$ is convex for all $r\ge 1$ and we consider $Kb$ as integer $>1$\footnote{For the case $Kb=1\Leftrightarrow K=1,b=1$, $\frac{x_{\min}}{x_{\max}}=1$, so that $g(a)=1,\forall a\in [0,1)$ trivially}. In $(\psi_4)$ the random variable $F:=\frac{X^2}{Z^2}$ is used, which has $F$-distribution, with parameters $(1,1)$, as it is a ratio of two i.i.d. $\chi^2$ random variables  with degrees of freedom $1$. Step $(\psi_5)$ uses the expression for the cumulative density function (cdf) of a $F$-distributed ranbdom variable, which involves the regularized incomplete beta function $I_{x}(a,b)$~\cite{abramowitz1964handbook} which is defined, for $a,b>0,\ x\ge 0$, as $I_{x}(a,b)=\frac{B(x;a,b)}{B(a,b)}$, where $B(x;a,b) = \int_0^1 t^{a-1}(1-t)^{b-1} dt$ is the incomplete beta function~\cite{pearson1948tables}, and $B(a,b)=B(1;a,b)$ is the beta function. Finally $(\psi_6)$ uses the fact that, $\forall\ x\ge 0$, $I_{x}(\frac{1}{2},\frac{1}{2})=\frac{2}{\pi}\sin^{-1}\sqrt{x}$, and that $\sin^{-1}\left(\frac{1}{1+a^2}\right)=\cot^{-1} a$, both of which can be easily verified.

Similarly, we can find an upper bound on $g$ as below:\begin{align}
\lefteqn{g(a)} & &\nonumber\\
\ & \stackrel{(\psi_7)}{\le} Kb\left(\expect{Z}{\expect{X}{{\indicator{Z<X<\frac{Z}{a}}}\right)^{Kb-1}}}\nonumber\\
\ & \stackrel{(\psi_{8})}{=} Kb\prob{1<F<\frac{1}{a^2}}\nonumber\\
\label{eq:g-function-upper-bound}
\ & = Kb\left(\frac{2}{\pi}\cot^{-1}a - \frac{1}{2}\right).
\end{align}
Here step $(\psi_7)$ again uses the convexity of the function $x^r,\ r\ge 1$, and step $(\psi_{8})$ uses the definition of the $F$-distributed random variable $F$ as defined before while deriving the inequality~\eqref{eq:g-function-lower-bound}.
\section{Finding upper bound on $N_{b,p,L'}(K,R)$}
\label{sec:appendix-upper-bound-combinatorial-N}
Let $S$ be a support for a PIBS vector with parameters $(b,p,L',L',K,R)$, with $(R+1)p\le K$. We will find the number of all such possible supports $S$.

The support $S$ correspond to the $k$ true clusters, with sizes $bj_1,\cdots,\ bj_k$, such that $1\le j_s\le p$, for all $s=1,\cdots,\ k$, and $\sum_{s=1}^k j_s=K$. Note that we must have $\lceil K/p\rceil \le k\le K$, so that $kp\ge K$. and between any two consecutive such clusters, there are at least $L'$ indices (which correspond to either some pseudoblocks or by zeros). The rest of the $n-Kb$ indices contain $R$ non-overlapping psedublocks scattered in these $k+1$ places in between the clusters and possibly beside the clusters around the edges(see the structure of a PIBS vector in Fig~\ref{fig:PIBS-structure-illustration}). Let the number of pseudoblocks in those $k+1$ places be denoted by $l_1,\cdots,l_{k+1}$, and the number of zeros in those places be denoted by $z_1,\cdots,z_{k+1}$. According to the PIBS structure, $l_1\cdots,l_{k+1}\in S_{1}$ and given such a tuple $(l_1\cdots,l_{k+1})$, denoted by $\bvec{l}$, $z_1,\cdots,z_{k+1}\in S_{2,\bvec{l}} $, where \begin{align}
S_1 & =\{\bvec{l}=(l_1\cdots,l_{k+1}):l_j\ge 0, 1\le j\le k+1, \sum_{j=1}^{k+1}l_j=R\},\nonumber\\
S_{2,\bvec{l}} & =\left\{z_1,\cdots,z_{k+1}: z_1\ge 0, z_{k+1}\ge 0,\ z_j\ge t_j,j=2\cdots,k;\right.\nonumber\\
\ & \left.\sum_{j=1}^{k+1}z_j=P\right\},
\end{align} where $P=n-Kb-L'R$, $t_j=(L'(1-l_j))^+,\ j=2,\cdots,k$, and $(x)^+=\max\{x,0\}$ for all $x\in \real$. Thus, given $k$ and $j_1,\cdots,j_k$, the number of ways $l_1,\cdots,l_{k+1},\ z_1,\cdots,z_{k+1}$ can be chosen so that they satisfy the above mentioned constraints, is given by $N_{j_1,\cdots,j_k,L'}(k,R) = \sum_{\bvec{l}\in S_1}\abs{S_{2,\bvec{l}}}$. Using the generating function technique~\cite[Chapter $1$]{stanley2011enumerative}, and using $kp\ge K$, it can be shown that $S_2(l_1,\cdots,l_{k+1})$ is the coefficient of $x^P$ in the expansion of the formal power series $(1+x+x^2+\cdots)\cdot\prod_{j=2}^k(x^{t_j}+x^{t_j+1}+\cdots)\cdot(1+x+x^2+\cdots)=x^{\sum_{j=2}^k t_j}(1-x)^{-(k+1)}$. Therefore, $\abs{S_2(l_1,\cdots,l_{k+1})}=\binom{k+P-\sum_{j=2}^k t_j}{k}$. Furthermore, observe that $\sum_{j=2}^k t_j=L'(k-1-q)$, where $q = \abs{Q}$, where $Q=\{j:2\le j\le k,\ l_j\ge 1\}$. Note that for a fixed $q$, there are $\binom{k-1}{q}$ such sets $Q$ where $l_j\ge 1$ for $j\in Q$. For each such choice of $Q$, the number of solutions of the equation $\sum_{j=1}^{k+1}l_j=R$, such that $l_1,l_{k+1}\ge 0$, and $l_j\ge 1$ for $j\in Q$, is the coefficient of $x^R$ in the expansion of $(1-x)^{-2}\cdot (x+x^2+\cdots)^{q}=x^q(1-x)^{-(q+2)}$. Hence, for a fixed $q,Q$, the total number of such $l_1,\cdots,l_{k+1}$ is $\binom{q+2+R-q-1}{R-q}=\binom{R+1}{R-q}$. Furthermore, note that the minimum value of $q$ can be $1$, when all the pseudoblocks are packed into one of the $k+1$ inter-cluster spaces. Moreover, the maximum value of $q$ can be $R$ since it is possible that all the $R$ pseudoblocks are accommodated among the available spaces in between the true clusters, since we have $R\le k-1$, as we have assumed that $R\le K/p-1\le k-1$. Thus, one obtains, \begin{align}
\lefteqn{N_{j_1,\cdots,j_k,L'}(k,R)} & & \nonumber\\
\label{eq:fixed-true-blocksizes-total-number-expression}
\ & = \sum_{q=1}^{R}\binom{k-1}{q}\binom{R+1}{R-q}\binom{k+P-L'(k-1-q)}{k}.
\end{align}
To find an upper bound on $N_{j_1,\cdots,j_k,L'}(k,R)$, we note that $\binom{k+P-L'(k-1-q)}{k}\le \binom{k+P-L'k+L'+L'R}{k}$, and use the identity $\sum_{j=0}^k \binom{m}{j}\binom{n-m}{k-j}=\binom{n}{k}$, to obtain \begin{align}
\lefteqn{ N_{j_1,\cdots,j_k,L'}(k,R)} & &\nonumber\\
\ & \le \binom{k+P-L'k+L'+L'R}{k}\left(\binom{R+k}{R}-(R+1)\right)\nonumber\\
\label{eq:fixed-true-blocksizes-total-number-upper-bound}
\ & \le e^k\left(\frac{L'+L'R+P}{k}-(L'-1)\right)^k\left(e^R(1+k/R)^R-(R+1)\right),
\end{align}
where in the last step we have used the inequality $\binom{n}{k}\le (en/k)^k$.

Now, note that this many solutions for the arrangements of the pseudoblocks is obtained for a fixed $k$ and $j_1,\cdots,j_k$, such that $1\le j_s\le p$ for $1\le s\le k$ and $\sum_{s=1}^k j_s=K$.
%
Now, note that the number of solutions to the equation $\sum_{s=1}^k j_s=K$ such that $1\le j_s\le p$, $1\le s\le k$, is equal to $p^k\prob{\sum_{s=1}^k X_s=K}$, where $X_s,\ 1\le s\le k$ are $k$ i.i.d. discrete random variables with $\prob{X_1=j}=\frac{1}{p}$, where $1\le j\le p$. We proceed to find upper bound on the quantity $\prob{\sum_{s=1}^k X_s=K}$.

First observe that $\expect{}{X_1}=\frac{p+1}{2}$, so that $\prob{\sum_{s=1}^k X_s=K}=\prob{\sum_{s=1}^k \left(X_s-\expect{}{X_s}\right)=K-\frac{p+1}{2}}$. If $K\ge \frac{p+1}{2}$, then, we write $\prob{\sum_{s=1}^k \left(X_s-\expect{}{X_s}\right)=K-\frac{p+1}{2}}\le \prob{\sum_{s=1}^k \left(X_s-\expect{}{X_s}\right)>K-\frac{p+1}{2}}$. Similarly, if $K<\frac{p+1}{2}$, we can then write, $\prob{\sum_{s=1}^k \left(X_s-\expect{}{X_s}\right)=K-\frac{p+1}{2}}\le \prob{\sum_{s=1}^k \left(X_s-\expect{}{X_s}\right)\le K-\frac{p+1}{2}}$. We now bound the probabilities $\prob{\sum_{s=1}^k \left(X_s-\expect{}{X_s}\right)>t}$, and $\prob{\sum_{s=1}^k \left(X_s-\expect{}{X_s}\right)\le t}$ when $t> 0$, and $t\le 0$ respectively.

When $t>0$, we can use the Hoeffding's inequality~\cite[Theorem~$2.8$]{boucheron2013concentration} to find \begin{align}
\label{eq:hoeffding-inequality-upper-bound}
\prob{\sum_{s=1}^k (X_s-\expect{}{X_s})>t} & \le \exp\left(-\frac{2t^2}{k(p-1)^2}\right),
\end{align}
which follows since each random variable $X_s$ is bounded in the interval $[1,p]$. Similarly, when $t<0,$ using the result of the problem $2.9$ of~\cite{boucheron2013concentration}, we find \begin{align}
\prob{\sum_{s=1}^k (X_s-\expect{}{X_s})\le t} & \le \exp\left(-\frac{t^2}{2\sum_{s=1}^k\expect{}{X_s^2}}\right)\nonumber\\
\label{eq:subgaussian-bound-nonnegative-randomn-variable}
\ & \le \exp\left(-\frac{3t^2}{k(p+1)(2p+1)}\right),
\end{align}
since $\expect{}{X_s^2}=\frac{(p+1)(2p+1)}{6}$. Thus, we find that \begin{align*}
\lefteqn{\prob{\sum_{s=1}^k \left(X_s-\expect{}{X_s}\right)=t}} & &\\
\ & \le\max\left\{\exp\left(-\frac{2t^2}{k(p-1)^2}\right),\exp\left(-\frac{3t^2}{k(p+1)(2p+1)}\right)\right\}\nonumber\\ \ & =\exp\left[-\frac{t^2/k}{\max\left\{\frac{(p-1)^2}{2},\frac{(p+1)(2p+1)}{3}\right\}}\right].
\end{align*}
Since $\frac{2(p+1)^2}{3}-\frac{(p-1)^2}{2}>\frac{(p+1)(2p+1)}{3}-\frac{(p-1)^2}{2}=\frac{p^2+12p-1}{6}>0$, we have that \begin{align}
\label{eq:preliminary-upper-bound-using-probability}
\prob{\sum_{s=1}^k X_s=k} < \exp\left(-\frac{3\left(K-k\frac{p+1}{2}\right)^2}{2k(p+1)^2}\right).
\end{align}
Furthermore, using $k\le K$, we deduce that $-\frac{\left(K-k\frac{p+1}{2}\right)^2}{k}=-\left(\frac{K^2}{k}+k\left(\frac{p+1}{2}\right)^2-K(p+1)\right)\le -\left(-Kp+k\left(\frac{p+1}{2}\right)^2\right).$ Thus, \begin{align}
\label{eq:refined-upper-bound-using-probability}
\prob{\sum_{s=1}^k X_s=k} < e^{\frac{3Kp}{2(p+1)^2}}\exp\left(-\frac{3k}{8}\right).
\end{align}
Hence, using~\eqref{eq:fixed-true-blocksizes-total-number-upper-bound} and~\eqref{eq:refined-upper-bound-using-probability}, we derive, \begin{align}
\lefteqn{N_{b,p,L'}(K,R)} & &\nonumber\\
\ & < \sum_{k=\lceil K/p \rceil}^K p^ke^{\frac{3Kp}{2(p+1)^2}}\exp\left(-\frac{3k}{8}\right)\nonumber\\
\ & \cdot e^k\left(\frac{L'+L'R+P}{k}-(L'-1)\right)^k\left(e^R(1+k/R)^R-(R+1)\right)\nonumber\\
\label{eq:total-arrangement-number-PIBS-vector-preliminary-upper-bound}
\ & \stackrel{(i)}{<} e^A\sum_{k=\lceil K/p\rceil}^K\exp\left[kC+k\ln\left(D/k-E\right)\right],
\end{align}
where $A=\frac{3Kp}{2(p+1)^2}+R,\ C=\ln p + \frac{13}{8},\ D=L'+L'R+P=n-Kb+L',\ E=L'-1$, and at step $(i)$ we have used the inequality $e^x\ge 1+x$, for $x\ge 0$, to obtain $e^R(1+k/R)^R-(R+1)<e^R(1+k/R)^R<e^{R+k}$. Note that the quantity $D/k-E$ is strictly greater than $1$ for all $\lceil K/p\rceil \le k\le K$. To verify this, note that as we have assumed that for a PIBS vector, $n\ge Kb+L'R+(K+ 1)(L'-1)$, we obtain $D-EK=n-Kb+L'-(L'-1)K\ge L'R+(K+1)(L'-1)+L'-(L'-1)K=L'R+2L'-1\ge 1$. Finally, noting that $\lfloor K/p \rfloor \le k\le K$, and using $1+x<e^x (x\ge 0)$, one can further upper bound $N_{b,p,L'}(K,R)$ as \begin{align}
\lefteqn{N_{b,p,L'}(K,R)} & & \nonumber\\
\ & < e^A \cdot \left(K-\lceil K/p\rceil+1\right)\exp\left[{KC}+K\ln(D/\lceil K/p\rceil -E)\right]\nonumber\\
\label{eq:total-arrangement-number-PIBS-vector-final-upper-bound}
\ & < \exp\left[A+KC'+K\ln\left(D'/K-E\right)\right],
\end{align}
where $C'=C+(1-1/p),\ D'=pD$.
\bibliography{two-stage-generalized-block-omp}

\begin{thebibliography}{10}
\providecommand{\url}[1]{#1}
\csname url@samestyle\endcsname
\providecommand{\newblock}{\relax}
\providecommand{\bibinfo}[2]{#2}
\providecommand{\BIBentrySTDinterwordspacing}{\spaceskip=0pt\relax}
\providecommand{\BIBentryALTinterwordstretchfactor}{4}
\providecommand{\BIBentryALTinterwordspacing}{\spaceskip=\fontdimen2\font plus
\BIBentryALTinterwordstretchfactor\fontdimen3\font minus
  \fontdimen4\font\relax}
\providecommand{\BIBforeignlanguage}[2]{{%
\expandafter\ifx\csname l@#1\endcsname\relax
\typeout{** WARNING: IEEEtran.bst: No hyphenation pattern has been}%
\typeout{** loaded for the language `#1'. Using the pattern for}%
\typeout{** the default language instead.}%
\else
\language=\csname l@#1\endcsname
\fi
#2}}
\providecommand{\BIBdecl}{\relax}
\BIBdecl

\bibitem{candes2006robust}
E.~J. Cand{\`e}s, J.~Romberg, and T.~Tao, ``Robust uncertainty principles:
  Exact signal reconstruction from highly incomplete frequency information,''
  \emph{IEEE Trans. Inf. Theory}, vol.~52, no.~2, pp. 489--509, 2006.

\bibitem{candes_decoding_2005}
E.~J. Cand{\`e}s and T.~Tao, ``Decoding by linear programming,'' \emph{IEEE
  Trans. Inf. Theory}, vol.~51, no.~12, pp. 4203--4215, Dec. 2005.

\bibitem{candes-tao-stable-recovery}
\BIBentryALTinterwordspacing
E.~J. Candès, J.~K. Romberg, and T.~Tao, ``Stable signal recovery from
  incomplete and inaccurate measurements,'' \emph{Comm. Pure Appl. Math.},
  vol.~59, no.~8, pp. 1207--1223, 2006. [Online]. Available:
  \url{http://dx.doi.org/10.1002/cpa.20124}
\BIBentrySTDinterwordspacing

\bibitem{cai2009recovery}
T.~T. Cai, G.~Xu, and J.~Zhang, ``On recovery of sparse signals via $l_1$
  minimization,'' \emph{IEEE Trans. Inf. Theory}, vol.~55, no.~7, pp.
  3388--3397, 2009.

\bibitem{donoho2012sparse}
D.~L. Donoho, Y.~Tsaig, I.~Drori, and J.-L. Starck, ``Sparse solution of
  underdetermined systems of linear equations by stagewise orthogonal matching
  pursuit,'' \emph{IEEE Trans. Inf. Theory}, vol.~58, no.~2, pp. 1094--1121,
  2012.

\bibitem{mallat_matching_1993}
S.~G. Mallat and Z.~Zhang, ``Matching pursuits with time-frequency
  dictionaries,'' \emph{IEEE Trans. Signal Process.}, vol.~41, no.~12, pp.
  3397--3415, Dec. 1993.

\bibitem{pati1993orthogonal}
Y.~C. Pati, R.~Rezaiifar, and P.~S. Krishnaprasad, ``Orthogonal matching
  pursuit: Recursive function approximation with applications to wavelet
  decomposition,'' in \emph{Signals, Systems and Computers, 1993. 1993
  Conference Record of The Twenty-Seventh Asilomar Conference on}.\hskip 1em
  plus 0.5em minus 0.4em\relax IEEE, 1993, pp. 40--44.

\bibitem{tropp2004greed}
J.~Tropp \emph{et~al.}, ``Greed is {G}ood: {A}lgorithmic {R}esults for {S}parse
  {A}pproximation,'' \emph{IEEE Trans. Inf. Theory}, vol.~50, no.~10, pp.
  2231--2242, 2004.

\bibitem{mishali2009blind}
M.~Mishali and Y.~C. Eldar, ``Blind multiband signal reconstruction: Compressed
  sensing for analog signals,'' \emph{IEEE Trans. Signal Process.}, vol.~57,
  no.~3, pp. 993--1009, 2009.

\bibitem{mishali2010theory}
------, ``From theory to practice: Sub-nyquist sampling of sparse wideband
  analog signals,'' \emph{IEEE J. Sel. Topics Signal Process.}, vol.~4, no.~2,
  pp. 375--391, 2010.

\bibitem{cotter2005sparse}
S.~F. Cotter, B.~D. Rao, K.~Engan, and K.~Kreutz-Delgado, ``Sparse solutions to
  linear inverse problems with multiple measurement vectors,'' \emph{IEEE
  Trans. Signal Process.}, vol.~53, no.~7, pp. 2477--2488, 2005.

\bibitem{chen2005sparse}
J.~Chen and X.~Huo, ``Sparse representations for multiple measurement vectors
  (mmv) in an over-complete dictionary,'' in \emph{Acoustics, Speech, and
  Signal Processing, 2005. Proceedings.(ICASSP'05). IEEE International
  Conference on}, vol.~4.\hskip 1em plus 0.5em minus 0.4em\relax IEEE, 2005,
  pp. iv--257.

\bibitem{eldar2009robust}
Y.~C. Eldar and M.~Mishali, ``Robust recovery of signals from a structured
  union of subspaces,'' \emph{IEEE Trans. Inf. Theory}, vol.~55, no.~11, pp.
  5302--5316, 2009.

\bibitem{parvaresh2008recovering}
F.~Parvaresh, H.~Vikalo, S.~Misra, and B.~Hassibi, ``Recovering sparse signals
  using sparse measurement matrices in compressed dna microarrays,'' \emph{IEEE
  J. Sel. Topics Signal Process.}, vol.~2, no.~3, pp. 275--285, 2008.

\bibitem{eldar2010block}
Y.~C. Eldar, P.~Kuppinger, and H.~Bolcskei, ``Block-sparse signals: Uncertainty
  relations and efficient recovery,'' \emph{IEEE Trans. Signal Process.},
  vol.~58, no.~6, pp. 3042--3054, 2010.

\bibitem{baraniuk2008simple}
R.~Baraniuk, M.~Davenport, R.~DeVore, and M.~Wakin, ``A simple proof of the
  restricted isometry property for random matrices,'' \emph{Constr. Approx.},
  vol.~28, no.~3, pp. 253--263, 2008.

\bibitem{wen2018sharp}
J.~Wen, Z.~Zhou, Z.~Liu, M.-J. Lai, and X.~Tang, ``Sharp sufficient conditions
  for stable recovery of block sparse signals by block orthogonal matching
  pursuit,'' \emph{Appl. Comput. Harmon. Anal.}, 2018.

\bibitem{gribonval2003harmonic}
R.~Gribonval and E.~Bacry, ``Harmonic decomposition of audio signals with
  matching pursuit,'' \emph{IEEE Trans. Signal Process.}, vol.~51, no.~1, pp.
  101--111, 2003.

\bibitem{zhang2013extension}
Z.~Zhang and B.~D. Rao, ``Extension of sbl algorithms for the recovery of block
  sparse signals with intra-block correlation,'' \emph{IEEE Trans. Signal
  Process.}, vol.~61, no.~8, pp. 2009--2015, 2013.

\bibitem{fang2015pattern}
J.~Fang, Y.~Shen, H.~Li, and P.~Wang, ``Pattern-coupled sparse bayesian
  learning for recovery of block-sparse signals,'' \emph{IEEE Trans. Signal
  Process.}, vol.~63, no.~2, pp. 360--372, 2015.

\bibitem{kannu2018spcom}
M.~A and A.~P. Kannu, ``Sinusoid signal estimation using generalized block
  orthogonal matching pursuit algorithm,'' in \emph{2018 International
  Conference on Signal Processing and Communications (SPCOM)}, Bangalore,
  India, 16-19 July 2018, pp. 60--64.

\bibitem{baraniuk2010model}
R.~G. Baraniuk, V.~Cevher, M.~F. Duarte, and C.~Hegde, ``Model-based
  compressive sensing,'' \emph{IEEE Trans. Inf. Theory}, vol.~56, no.~4, pp.
  1982--2001, 2010.

\bibitem{tropp2007signal}
J.~A. Tropp and A.~C. Gilbert, ``Signal {R}ecovery {F}rom {R}andom
  {M}easurements {V}ia {O}rthogonal {M}atching {P}ursuit,'' \emph{IEEE Trans.
  Inf. Theory}, vol.~53, no.~12, pp. 4655--4666, 2007.

\bibitem{wang2017recovery}
J.~Wang and P.~Li, ``Recovery of {S}parse {S}ignals {U}sing {M}ultiple
  {O}rthogonal {L}east {S}quares,'' \emph{IEEE Trans. Signal Process.},
  vol.~65, no.~8, pp. 2049--2062, April 2017.

\bibitem{foucart2013mathematical}
S.~Foucart and H.~Rauhut, \emph{A mathematical introduction to compressive
  sensing}.\hskip 1em plus 0.5em minus 0.4em\relax Springer, 2013.

\bibitem{dai2009subspace}
W.~Dai and O.~Milenkovic, ``Subspace pursuit for compressive sensing signal
  reconstruction,'' \emph{IEEE Trans. Inf. Theory}, vol.~55, no.~5, pp.
  2230--2249, 2009.

\bibitem{davenport2010analysis}
M.~Davenport, M.~B. Wakin \emph{et~al.}, ``Analysis of orthogonal matching
  pursuit using the restricted isometry property,'' \emph{IEEE Trans. Inf.
  Theory}, vol.~56, no.~9, pp. 4395--4401, 2010.

\bibitem{wen2017novel}
J.~Wen, Z.~Zhou, D.~Li, and X.~Tang, ``A novel sufficient condition for
  generalized orthogonal matching pursuit,'' \emph{IEEE Commun. Lett.},
  vol.~21, no.~4, pp. 805--808, 2017.

\bibitem{wen2017nearly}
J.~Wen, J.~Wang, and Q.~Zhang, ``Nearly optimal bounds for orthogonal least
  squares,'' \emph{IEEE Trans. Signal Process.}, vol.~65, no.~20, pp.
  5347--5356, 2017.

\bibitem{cai2011orthogonal}
T.~T. Cai and L.~Wang, ``Orthogonal matching pursuit for sparse signal recovery
  with noise,'' \emph{IEEE Trans. Inf. Theory}, vol.~57, no.~7, pp. 4680--4688,
  2011.

\bibitem{horn1990matrix}
R.~A. Horn, R.~A. Horn, and C.~R. Johnson, \emph{Matrix analysis}.\hskip 1em
  plus 0.5em minus 0.4em\relax Cambridge university press, 1990.

\bibitem{szarek1991condition}
S.~J. Szarek, ``Condition numbers of random matrices,'' \emph{J. Complex.},
  vol.~7, no.~2, pp. 131--149, 1991.

\bibitem{davidson2001local}
K.~R. Davidson and S.~J. Szarek, ``Local operator theory, random matrices and
  banach spaces,'' \emph{Handbook of the geometry of Banach spaces}, vol.~1,
  no. 317-366, p. 131, 2001.

\bibitem{abramowitz1964handbook}
M.~Abramowitz and I.~Stegun, ``Handbook of mathematical functions: With
  formulas, graphs, and mathematical tables applied mathematics series,''
  \emph{National Bureau of Standards, Washington, DC}, 1964.

\bibitem{pearson1948tables}
K.~Pearson, \emph{Tables of the Incomplete Beta-function}.\hskip 1em plus 0.5em
  minus 0.4em\relax Biometrika, 1948.

\bibitem{stanley2011enumerative}
R.~P. Stanley, ``Enumerative combinatorics (volume 1 second edition),''
  \emph{Cambridge studies in advanced mathematics}, 2011.

\bibitem{boucheron2013concentration}
S.~Boucheron, G.~Lugosi, and P.~Massart, \emph{Concentration inequalities: A
  nonasymptotic theory of independence}.\hskip 1em plus 0.5em minus 0.4em\relax
  Oxford university press, 2013.

\end{thebibliography}
\end{document}